\documentclass[11pt]{article}

\usepackage{times}
\usepackage{url, epsfig, amsmath, amsthm, amssymb, graphicx}
\usepackage{caption}
\usepackage{wrapfig}
\usepackage[usenames,dvipsnames]{color}
\usepackage[aboveskip=-50pt]{subcaption}

\newtheorem*{corollary*}{Corollary}

\DeclareMathOperator*{\argmin}{arg\,min}

\def\F{\ensuremath{\mathcal{F}}}
\def\H{\ensuremath{\mathcal{H}}}
\def\R{\ensuremath{\mathcal{R}}}
\def\C{\ensuremath{\mathcal{C}}}

\def\Q{\ensuremath{\mathcal{Q}}}
\def\P{\ensuremath{\mathcal{P}}}
\def\S{\ensuremath{\mathcal{S}}}
\def\O{\ensuremath{\mathcal{O}}}

\def\I{\ensuremath{\mathcal{I}}}
\def\opt{{\textsc{Opt}}}
\def\approx{{\textsc{Approx}}}

\def\T{\ensuremath{\mathcal{T}}}

\def\etal{\emph{et al.}}
\newcounter{listcounter}

\renewcommand{\thelistcounter}{\roman{listcounter}}
\newcommand{\descr}{\begin{list}{(\thelistcounter)}
{\usecounter{listcounter}
\setlength{\rightmargin}{0mm}}}
\newtheorem{lemma}{Lemma}[section]
\newtheorem{theorem}[lemma]{Theorem}

\newtheorem{claim}[lemma]{Claim}

\newif\ifpfsymb

\newenvironment{Pf}{\par\addvspace{6pt}\addtocounter{proof}{1}
       \noindent{\bf Proof:}}{{\ifnum\value{proof}>0
                \hfill\hbox{\mbox{}\hfill$\Box$}
		\addtocounter{proof}{-1}\par\addvspace{10pt}\else\par\fi
        }}

\newenvironment{defn}{\par\addvspace{13pt}\noindent{\bf Definition
		\addtocounter{lemma}{1} \hspace{-1mm}\thelemma
		\hspace{2mm}}}{\par\addvspace{10pt}}

\makeatletter
\newcounter{algo}
\def\thealgo{\@arabic\c@algo}
\def\fps@algo{tbp}
\def\ftype@algo{1}
\def\ext@algo{loa}
\def\fnum@algo{Algorithm \thealgo}
\def\algo{\@float{algo}}
\def\endalgo{\end@float}
\makeatother

\makeatletter
\def\remlab#1{\@bsphack\if@filesw {\let\thepage\relax
   \def\protect{\noexpand\noexpand\noexpand}%
\xdef\@gtempa{\write\@auxout{\string
	   \newlabel{rem:#1}{{\thelemma}{\thepage}}}}}\@gtempa
            \if@nobreak \ifvmode\nobreak\fi\fi\fi\@esphack}
\def\deflab#1{\write\@auxout{\string
	\newlabel{def:#1}{{\thelemma}{\thepage}}}}
\makeatother

{\catcode`\@=11
\gdef\setft#1#2#3{%
\def\@oddfoot{%
{\setbox0=\hbox{#1}%
\setbox1=\hbox{#3}%
\ifdim\wd0>\wd1%
\dimen0=\wd0%
\box0\hfil#2\hfil\hbox to\dimen0{\hfil\hfil\box1}%
\else \dimen0=\wd1%
\hbox to\dimen0{\box0\hfil}\hfil#2\hfil\box1\fi%
}}}}

{\catcode`\@=11
\gdef\sethd#1#2#3{%
\def\@oddhead{\vbox{\hbox to\hsize{{#1}\hfil{#2}\hfil{#3}}%
\vspace{0.06in}%
\hbox to \hsize{\hrulefill}\vspace*{-0.09in}}}
\def\@evenhead{\@oddhead}
	}
}

\def\mysecn#1{\setcounter{equation}{0}
\section*{#1}\mark{#1}}

\def\thebibliography#1{\mysecn{References}
\addcontentsline{toc}{section}{References}\list
{[\arabic{enumi}]}{\settowidth\labelwidth{[#1]}\leftmargin\labelwidth
 \advance\leftmargin\labelsep
 \usecounter{enumi}}
 \def\newblock{\hskip .11em plus .33em minus .07em}
 \sloppy\clubpenalty4000\widowpenalty4000
 \sfcode`\.=1000\relax
 \small}

\def\complaint#1{}
\def\withcomplaints{
\newcounter{mycomplaints}
\def\complaint##1{\refstepcounter{mycomplaints}%
\ifhmode%
\unskip%
{\dimen1=\baselineskip \divide\dimen1 by 2 %
\raise\dimen1\llap{\tiny -\themycomplaints-}}\fi%
\marginpar{\tiny [\themycomplaints]: ##1}}%
}

\newcounter{printertype}
\def\figprint#1{
        \ifcase \theprintertype

		\begin{center}
                 \input{#1}
		\end{center}
              \or
                 \centerline{\psfig{figure=#1.ps}}
              \else
                 \vspace*{1in}
        \fi}

\makeatletter
\long\def\@myfootnotetext#1{\insert\footins{\footnotesize
    \interlinepenalty\interfootnotelinepenalty 
    \splittopskip\footnotesep
    \splitmaxdepth \dp\strutbox \floatingpenalty \@MM
    \hsize\columnwidth \@parboxrestore
   \edef\@currentlabel{\csname p@footnote\endcsname\@thefnmark}\@makemyfntext
    {\rule{\z@}{\footnotesep}\ignorespaces
      #1\strut}}}

\def\myfootnotetext{\@ifnextchar
[{\@xfootnotenext}{\xdef\@thefnmark{\thempfn}\@myfootnotetext}}

\long\def\@makemyfntext#1{\parindent 5mm #1}

\newcounter{proof}
\def\@@meqncr{\let\@tempa\relax
    \ifcase\@eqcnt \def\@tempa{& & &}\or \def\@tempa{& &}
      \else \def\@tempa{&}\fi
     \@tempa $\Box$\addtocounter{proof}{-1}
     \global\@eqnswtrue\global\@eqcnt\z@\cr}
\@namedef{meqnarray*}{\def\@eqncr{\nonumber\@seqncr}\eqnarray}
\@namedef{endmeqnarray*}{\global\@eqcnt\tw@\@@meqncr\egroup
      \global\advance\c@equation\m@ne$$\global\@ignoretrue}

\@namedef{meqnarray}{\eqnarray}
\@namedef{endmeqnarray}{\@@meqncr\egroup
        \global\advance\c@equation\m@ne$$\global\@ignoretrue}
\def\mequation{$$\global\@ignoretrue}

\@namedef{Pf*}{\Pf}
\@namedef{endPf*}{\par\addvspace{10pt}}

\makeatother

\input
\def\F{\ensuremath{\mathcal{F}}}
\def\H{\ensuremath{\mathcal{H}}}
\def\R{\ensuremath{\mathcal{R}}}
\def\C{\ensuremath{\mathcal{C}}}
\def\eps{\ensuremath{\epsilon}}

\def\d2{\ensuremath{\lfloor d/2 \rfloor}}

\newcommand{\remove}[1]{}

\newcommand{\hide}[1]{}

\newtheorem{observation}{Observation}[section]
\newtheorem{defi}{Definition}[section]

\renewcommand{\Re}{\mathbb{R}}

\newsavebox{\smallProofsym}                            
\savebox{\smallProofsym}                               %
{
\begin{picture}(6,6)                                   %
\put(0,0){\framebox(6,6){}}                            %
\put(0,2){\framebox(4,4){}}                            %
\end{picture}                                          %
}   

\begin{document}

\title{QPTAS for Geometric Set-Cover Problems via Optimal Separators} 
\date{}
 
\author{ 
        Nabil H. Mustafa \\ Universit\'e Paris-Est, \\ Laboratoire d'Informatique Gaspard-Monge, \\ Equipe A3SI, ESIEE Paris.\\ mustafan@esiee.fr \and
        Rajiv Raman \\ Dept. of Computer Science. \\  IIIT, Delhi.  \\ rajiv@iiitd.ac.in \and
        Saurabh Ray \\  Computer Science, New York University, Abu Dhabi. \\ saurabh.ray@nyu.edu}


\maketitle
\thispagestyle{empty}

\begin{abstract}
\hide{
\textbf{FIX SUBTLETY WITH INSIDE/OUTSIDE SEPARATOR}: no point on the separator
boundary\\
-LEMMAS IN THE MAIN THAT REFER TO THE LEMMAS IN THE APPENDIX, FIX THAT \\
}

Weighted geometric set-cover problems arise naturally in several   geometric
and non-geometric settings
(e.g. the breakthrough of Bansal-Pruhs (FOCS 2010) reduces a wide class of machine
scheduling problems to weighted geometric set-cover).
More than two decades of research has succeeded in settling
the $(1+\eps)$-approximability status for most geometric set-cover problems, except
for four basic scenarios which are still lacking. 
One is that of weighted
disks in the plane for which, after a series of papers, 
Varadarajan (STOC 2010) presented a clever
\emph{quasi-sampling} technique, which together
with improvements by Chan \etal~(SODA 2012), yielded
a $O(1)$-approximation algorithm.
Even for the unweighted case, a PTAS for a fundamental class
of objects called pseudodisks (which includes disks, unit-height rectangles,
translates of convex sets etc.) is currently unknown. 
Another fundamental case is weighted halfspaces in $\Re^3$, for which
a PTAS is currently lacking.
In this paper, we present a QPTAS for all of these remaining problems. 
Our results are based
on the separator framework of Adamaszek-Wiese (FOCS 2013, SODA 2014), who recently obtained a QPTAS for weighted independent set of polygonal regions.
This rules out the possibility that these problems are APX-hard,
assuming $\textbf{NP} \nsubseteq \textbf{DTIME}(2^{polylog(n)})$.
Together with the recent work of Chan-Grant (CGTA 2014), this  settles the APX-hardness status  
for all natural geometric set-cover problems.

\end{abstract}
 
\newpage

\setcounter{page}{1}

\section{Introduction}
\label{sec:introduction}

One of the fundamental optimization problem is
 the set-cover problem:
 given
a range space $(X, \R)$ consisting of a set $X$ and 
a set $\R$ of subsets of $X$ called the \emph{ranges},   the
objective is to compute a minimum-sized subset of $\R$ that
covers all the points of $X$. 
Unfortunately in the general case, it is strongly NP-hard;
worse, it is NP-hard to approximate the minimum set-cover
within a factor of $c\log n$ of the optimal~\cite{RS97} for some constant $c$.

A natural extensively-studied occurrence of the set-cover problem is when
the ranges are derived from geometric objects. For example,
given a set $P$ of $n$ points in the plane and a set
$\R$ of disks, the set-cover
problem for disks asks to compute a minimum cardinality subset of disks 
whose union covers all the points of $P$. 
Unfortunately, computing the minimum cardinality set-cover
remains NP-hard even for basic geometric objects, such
as unit disks in the plane.  
Effort has therefore been devoted to devising approximation algorithms
for geometric set-cover problems (see~\cite{CKL07, NV06, GMWZ04,  CV07, AES12, CF13, CDDFLS09, HM84, HP2009, NV06, CG2014, ELJ08, DFLN11, CGKS12, MR10} for a few examples).
Nearly all the effort has been for the following natural and fundamental
categories of geometric objects: \emph{halfspaces}, \emph{balls} (and
 generally, \emph{pseudodisks}), 
\emph{axis-parallel rectangles}, \emph{triangles} and 
\emph{objects parameterized by their union-complexity} (a set of regions $\mathcal{R}$ has union complexity $\phi(\cdot)$ if the boundary of the union of any $r$ of the regions has at most $r\phi(r)$ intersection points).
An important version is the \emph{weighted} setting,
where one seeks to find the minimum-weight set-cover
(e.g., see the breakthrough of Bansal-Pruhs~\cite{BP10} who reduced a broad class of machine scheduling problems to weighted set-cover problems).

Research during the past three decades has, in fact, been
able to largely answer the question of the existence
of a PTAS, or provability of APX-hardness for these problems
for the uniform case, where one is minimizing the cardinality
of the set-cover.  For the more general weighted case, 
there has been considerable progress recently --
a $O(\log \phi(\opt))$-approximation as a function of the union-complexity
 is  possible via the quasi-uniform
sampling technique of Varadarajan~\cite{V10} and its improvement by Chan \emph{et al.}~\cite{CGKS12}.
On the other hand, recently Chan-Grant~\cite{CG2014} proved
APX-hardness results for the set-cover problem
for a large class of geometric objects.
We also point out (Appendix~\ref{appendix:lowerbounds}) that for any integer $s$, there exist $O(s)$-sided
polygons with union complexity $n 2^{\alpha(n)^{O(s)}}$ for which
set-cover is inapproximable within $\Omega(\log s)$. 
Also, as any set-system with sets of size at most $s$ can be realized
by halfspaces in $\Re^{2s}$, a $\Omega(\log d)$ lower-bound follows
for approximability of halfspaces set-cover in $\Re^d$ (Appendix~\ref{appendix:lowerbounds}). This lower-bound
requires $d \geq 4$, leaving open the interesting question
of approximation schemes for weighted halfspaces in $\Re^3$.

\begin{wrapfigure}{r}{0.42 \textwidth}
\vspace{-0.2in}
\begin{tabular}{|l|c|c|}
\hline
\textbf{Object} & \textbf{Uniform}  & \textbf{Weighted} \\ \hline
Halfspaces ($\Re^2$) & Exact & Exact \\ \hline
Halfspaces ($\Re^3$)& PTAS & \textcolor{red}{?} \\ \hline
Halfspaces ($\Re^4$)& APX-H & APX-H \\ \hline
Balls ($\Re^2$)& PTAS & \textcolor{red}{?} \\ \hline
Pseudodisks ($\Re^2$)& \textcolor{red}{?} & \textcolor{red}{?} \\ \hline
Balls ($\Re^3$) & APX-H & APX-H \\ \hline
A-P Rects ($\Re^2$)& APX-H & APX-H \\ \hline
A-P Rects ($\Re^d$)& APX-H & APX-H \\ \hline
Triangles ($\Re^2$)& APX-H & APX-H \\ \hline
Linear U-C ($\Re^2$) & APX-H & APX-H \\ \hline
\end{tabular}
\vspace{-0.2in}
\end{wrapfigure}

See the table for the current status of geometric set-cover.
The four open cases present a challenge as the current state-of-the-art methods
hit some basic obstacles:
the approximation algorithms
for weighted halfspaces, balls and pseudodisks use
LP-rounding with $\epsilon$-nets~\cite{V10,CGKS12}, 
and so provably cannot give better than $O(1)$-approximation
algorithms. LP rounding was avoided by 
the use of local-search technique~\cite{MR10} to give a
 PTAS for halfspaces ($\Re^3$) and disks ($\Re^2$); however $i)$
for fundamental reasons it is currently limited to the unweighted case,
and $ii)$ does not extend to pseudodisks.

In this paper, we make progress on the approximability
status of the remaining four open cases by presenting
a quasi-polynomial time approximation scheme (QPTAS) for all these problems. 
This rules out the possibility that these problems are APX-hard,
assuming $\textbf{NP} \nsubseteq \textbf{DTIME}(2^{polylog(n)})$.
Together with the previous work  showing
hardness results or PTAS, this   settles the APX-hardness status  
for all natural geometric set-cover problems.

%
%
%
%
%


The motivation of our work is the recent progress on approximation
algorithms for another fundamental geometric optimization problem, maximum
independent sets in the intersection graphs of geometric objects,
where  $(1+\eps)$-approximation algorithms (or even constant factor approximation algorithms)
are not known for many objects.
In a recent breakthough, Adamaszek-Wiese~\cite{AW13, AW14} presented a 
QPTAS
for computing weighted maximum independent set for a variety of geometric
objects (e.g., axis-parallel rectangles, line-segments, polygons  with
polylogarithmically many sides) in the plane (the algorithm runs in time $2^{poly(\log n/\eps)}$).
We now sketch their main idea for approximating
the maximum independent set for weighted line-segments
in the plane, for which let $\opt$ be the optimal solution.
The key tool is the existence
of a closed polygonal curve $\C$ (with few vertices) that intersects segments in $\opt$ with small total weight, and at least a
constant fraction of the total weight of $\opt$ lies in the
two regions created by $\C$. Hence one can
guess the curve $\C$ (which does not require knowing $\opt$)~\footnote{The guessing is actually done by enumerating all possible curves. The fact that $\C$ has a small number of vertices allows efficient enumeration.}, and then return  the union
of solution of the  two sub-problems  (which are solved recursively).
With appropriate parameters, the loss incurred by throwing away the segments
intersecting $\C$ is at most $\eps$-th fraction of the optimal solution, yielding
a $(1-\eps)$-approximation in quasi-polynomial time.

Let us consider how the above technique can be made to work for the set cover problem. Assume that
we are given a set $P$ of $n$ points and a set $\R$ of $m$ weighted disks and our goal is to pick
a minimum-weight set-cover from $\R$.
We can again consider the optimal solution \opt{} and hope to find a curve $\C$ which intersects objects 
in \opt{} with small total weight and has a constant fraction of the weight of \opt{} in the interior as well
as the exterior. However such a curve does not always exist -- consider, e.g., a case where the optimal solution consists of a set of disks that share a common point (not necessarily in $P$).
Crucially, unlike the independent set problem, the objects in the optimal set-cover are not disjoint.
This dooms any separator-based approaches for the set-cover problem.

Surprisingly, we show that nevertheless there still exists a curve $\C$ 
(which may, in fact, intersect all the disks in $\R$!) such   that solving the induced
sub-problems in the interior and exterior of $\C$ and combining them  leads 
to a near-optimal solution (Theorem~\ref{lem:rajnikant}).
The problem is further complicated by arbitrary weights on the disks. 
As a result, several promising approaches (including the quasi-sampling
technique of Varadarajan~\cite{V10}) fail. 
Fortunately, generalizing the problem to pseudodisks and then using structural properties of pseudodisks melded 
with randomized ordering and probabilistic re-sampling techniques works out.

Such separator based techniques do not work in three dimensions 
(even for the independent set problem, one can show that there exists 
a set of disjoint segments in $\Re^3$ so that there is no compactly-represented polyhedral separator). 
In fact, even for unit balls in $\Re^3$ all containing a common point, 
the set cover problem is APX-hard~\cite{CG2014}! 
However, when the objects are halfspaces in $\Re^3$, we prove the existence of a polyhedral separator that 
allows us to get a QPTAS. This shows that the set-cover
problem for halfspaces is the only
natural problem in $\Re^3$ that is not APX-hard.
 

\section{Preliminaries}

Let $\R = \{R_1, \ldots, R_n\}$ be a set of weighted $\alpha$-simple regions in the plane, where 
a bounded and connected region in the plane is called $\alpha$-simple if its boundary can be decomposed into at most $\alpha$ $x$-monotone arcs~\footnote{Note that whether a region is $\alpha$-simple depends on choice of axes.}.
For any $\alpha$-simple region $R$, we denote by ${\Gamma}(R)$ a set of at most $\alpha$ $x$-monotone curves that its boundary can be 
decomposed into. For a set $\mathcal{R}$ of $\alpha$-simple regions, we define $\Gamma(\mathcal{R})$ to be the set $\bigcup_{R\in \mathcal{R}} \Gamma(R)$. 
Let $w_i$ denote the weight of the region $R_i \in \R$,
and $w(\S)$ be the total weight of the regions in $\S$ (set $W = w(\R)$). The regions in $\R$ need not
be disjoint~\footnote{We assume $\R$ to be in general position, so no
three regions boundaries intersect at the same point.}. 
A collection of compact simply connected regions in the plane is said to form a family of \emph{pseudodisks} if the boundaries of any two of the regions intersect at most twice. 
The union complexity of a set of pseudodisks in linear~\cite{PS09}.
For technical reasons we will assume that pseudodisks in this paper are $\alpha$-simple for some constant $\alpha$. This restriction is not crucial,
and can be removed~\cite{SHP14}. 
A collection of regions $\R$ is said to be {\em cover-free} if no region $R \in \R$ is covered by the union of the regions
in $\R \setminus R$.
For any closed Jordan curve $\C$, 
we denote the closed region
bounded by it as $interior(\C)$ and the closed 
unbounded region defined by it as 
$exterior(\C)$. Given $\R$, we denote by $\R_{in}(\C)$ ($\R_{ext}(\C)$) the subset of the regions that lie
in $interior(\C)$ ($exterior(\C)$). Similarly if $P$ is a set of points, we denote by $P_{in}(\C)$ ($P_{ext}(\C)$) the subset of points lying in $interior(\C)$ ($exterior(\C)$).\\


{\bf VC-dimension and $\eps$-nets~\cite{M02}.} Given a range space $(X, \F)$, a set $X' \subseteq X$ is \emph{shattered} if every subset
of $X'$ can be obtained by intersecting $X'$ with a member of the family $\F$. The VC-dimension of $(X,\F)$ 
is the size of the largest set that can be shattered. 
Given a set system $(X, \F)$ where each element of $X$ has a positive weight associated with it, and a paramter $0 <\epsilon <1$, an $\epsilon$-net is a subset $Y \subseteq X$ s.t. for any $F \in \F$ with weight at least an $\epsilon$ fraction of the total weight, $Y \cap F \neq \emptyset$. The $\eps$-net 
theorem (Haussler-Welzl~\cite{HW87})
states that there exists an $\eps$-net of size $O(d/\eps \log 1/\eps)$ for any range space with VC-dimension $d$.  \\


{\bf QPT-partitionable problems.} Given an optimization problem $\O$,
let $\opt_{\O}(I)$ denote the optimal solution
of $\O$ on the instance $I$, and let 
$w(\opt_{\O}(I))$ be the weight of this optimal solution. 
We assume $\O$ is a minimization problem; similar statements
hold for the maximization case.

\begin{defi}
A problem $\O$ is \emph{quasi-polynomial time partitionable} 
(\emph{QPT-partitionable})
if, given any input $I$ and a parameter $\delta > 0$, there exist a constant $c < 1$, 
$k = O\left(n^{(\log{n}/\delta)^{O(1)}}\right)$, and
instance pairs $(I^1_l, I^1_r), \ldots, (I^k_l, I^k_r)$ (computable in time polynomial in $k$),
and an index $j$, $1 \leq j \leq k$, such that 
$i)$ $\max \{ w(\opt_{\O} (I^j_l)), w(\opt_{\O}(I^j_r)) \} \leq c \cdot w(\opt_{\O}(I))$,
$ii)$ $\opt_{\O}(I^j_l) \cup \opt_{\O}(I^j_r)$ is a feasible solution, and
$iii)$ $w(\opt_{\O}(I^j_l))+w(\opt_{\O}(I^j_r)) \leq (1+\delta)w(\opt_{\O}(I))$.\end{defi}
The next lemma  follows immediately from recursive divide-and-conquer:
\begin{lemma}
\label{thm:alg}
If a problem $\O$ is QPT-partitionable, and 
if for any instance $I$, $w(\opt_{\O}(I)) \geq 1$,  then
one can compute a $(1+\eps)$-approximate solution for
$\O$ in time $O\left(n^{(\frac{1}{\eps} \cdot \log w(\opt_{\O}(I)) \cdot \log n)^{O(1)}}\right)$.
\end{lemma}
\begin{proof}
The algorithm will return an approximate solution $\approx_{\O}(I)$ as follows.
Let $T = w(\opt_{\O}(I))$, and set $\delta = \Theta( \eps / \log T )$.
Construct the $k$ instance pairs $(I^1_l, I^1_r), \ldots, (I^k_l, I^k_r)$,
where $k = O\left(n^{(1/\eps \cdot \log T \cdot \log n)^{O(1)}}\right)$.
For each $i = 1 \ldots k$, 
compute $\approx_{\O}(I^i_l)$ and $\approx_{\O}(I^i_r)$ recursively and return
the solution $\approx_{\O}(I^j_l) \cup \approx_{\O}(I^j_r)$,
where  $j = \argmin_i w(\approx_{\O}(I^i_l))+w(\approx_{\O}(I^i_r))$.
We can prune the recursion tree at the level 
$l = O(\log T)$ since 
for the  right  choice of $i$ at each recursion, the weight of the optimal solution
falls by a constant factor with every recursive call. The size of the tree is at most 
$(2k)^l = O\left(n^{(1/\eps \cdot \log T \cdot \log n)^{O(1)}}\right)$.  It can be shown inductively that the approximation factor of a sub-problem $t$ levels away from the lowest level is $(1+\delta)^t$. Thus the approximation factor at the root is $(1+\delta)^l \leq (1+\epsilon)$, with appropriate
constants in the definition of $\delta$. The time taken by the algorithm is $O\left(n^{(1/\eps \cdot \log T \cdot \log n)^{O(1)}}\right)$. 
\end{proof}


\paragraph{Geometric separators.} 
A $\delta$-separator for $\R$, given $\delta > 0$, is a simple closed curve $\C$ 
in the plane such that the number of regions of $\R$ completely 
inside (and outside) $\C$ is at most $2W/3$ (such a curve is called \emph{balanced}),
and the total weight of the regions in $\R$ intersecting $\C$ is 
at most $\delta W$. The goal
is to show the existence, given $\R$ and $\delta > 0$, of separators 
of small combinatorial complexity as a function of $n$, $m$ (number of intersections
in $\R$), $\alpha$ and $\delta$. 
The existence of small $\delta$-separators was the core of the result of~\cite{AW13, AW14}; later 
several authors noted~\cite{MRR13, SHP14} that the construction in~\cite{AW13}
can be made optimal using the 
techniques of constructing cuttings and $\eps$-nets (i.e, the probabilistic
re-sampling technique)~\cite{CF90, M02, CV07, AES10}.

We state two separator results that we will be using in our algorithm.
For completeness we outline their proof in Appendix~\ref{appendix:separators}.

\begin{theorem}[~\cite{MRR13,SHP14}]
\label{thm:weightedseparator}
Given a set $\R$ of $n$ weighted 
 $\alpha$-simple regions (with total weight $W$, and no curve
having weight more than $W/3$) with disjoint interiors, and a parameter $\delta>0$, 
there exists a simple closed curve $\C$ such that $i)$ the total weight of the regions intersecting
$\C$ is at most $\delta W$, and $ii)$ the total weight of the regions
completely inside or outside $\C$ is at most $2W/3$. Furthermore the complexity of 
$\C$ is $T =  O(\alpha/\delta)$. 
 That is $\C$ can be completely described by a sequence of at most $T$ curves of $\Gamma(\mathcal{R})$
and additional at most $T$ bits.
Furthermore this is optimal; even when $\R$ is set
of disjoint line segments, any $\C$
satisfying these two properties must have $\Omega(1/\delta)$ bends.
\end{theorem}

In the case when the regions have uniform
weights (say each region has weight one) but are not necessarily disjoint:

\begin{theorem}[~\cite{MRR13}]
\label{thm:separator}
Given a set $\R$ of $n$ 
$\alpha$-simple regions in the plane with $m$ intersections,
and a parameter $r$,
there exists a simple closed curve $\C$ such that
$i$) the number of regions in $\R$ intersecting 
$\C$ are $O ( (m + \frac{\alpha^2 n^2}{r})^{1/2} )$, and
$ii$) the total number of regions
completely inside or outside $\C$ is at most $2n/3$.
Furthermore, complexity of $\C$ is $T = O( (r + \frac{mr^2}{\alpha^2 n^2})^{1/2} )$. 
That is,
$\C$ can be completely described by a sequence
of at most $T$ curves
of $\Gamma(\R)$ and at most $T$ additional bits.
\end{theorem}
The technical condition that the regions are $\alpha$-simple in the theorems above can be removed~\cite{SHP14}.
\paragraph{Geometric Set Cover.} 
Let $\R = \{R_1, \ldots, R_n\}$ be a set of weighted regions (in $\Re^2$
or $\Re^3$) and let $P$ be a finite set of points in the plane. The goal is to compute a subset $\mathcal{Q} \subseteq \mathcal{R}$ minimizing the total weight $w(\mathcal{Q})$ so that $P\subseteq \cup_{Q\in \mathcal{Q}}Q$. We will denote an optimal solution $\mathcal{Q}$ for an instance of the problem given by a set of regions $\mathcal{R}$ and a set of points $P$ by $\textsc{Opt}(\mathcal{R}, P)$,
and its weight by $w(\opt(\R, P))$.

\begin{claim}
\label{claim:tofame}
If there exists a QPTAS for set-systems $(\R', P')$ 
where 
$i)$ each $R \in \R'$ has weight $w(R) \geq 1$, and
$ii)$ the weight
of the optimal set-cover for $(\R', P')$ is $O(n/\eps)$, 
then there exists a QPTAS for the minimum-weight set cover
for a set-system $(\R, P)$ with arbitrary weights.
%
\end{claim}
\begin{proof}
Let $\Q$ be a minimum-weight set-cover for $(\R, P)$. 
First guess the maximum weight region
in $\Q$, say of weight $w_{max}$ (there are $n$ such choices). Then by exponential search on the 
interval $[w_{max}, nw_{max}]$, one can guess the weight of $\Q$ within a $(1+\eps/3)$ factor (there are $O(\log_{1+\eps}n)$ such
choices). Let $w_{aprx}$ be this weight, satisfying
$w(\Q) \leq w_{aprx} \leq (1+\eps/3)w(\Q)$.
Set $\R' \subset \R$ to be the set of regions with
weight at least $\eps w_{aprx}/n$, and $\R'' = \R \setminus \R'$.
Let $P' \subseteq P$ be the set of points not covered by $\R''$,
and construct a $(1+\eps/3)$ approximate set-cover $\Q'$
to $(\R', P')$. Return $\Q' \cup \R''$ as a set-cover
for $(\R, P)$. Note that this is the required approximation:
\begin{eqnarray*}
w(\Q' \cup \R'') = w(\Q') + w(\R'') &\leq& (1+\eps/3)w(\opt(\R', P')) + \eps w_{aprx} \\
& \leq &(1+\eps/3) w(\Q) + \eps (1+\eps/3) w(\Q) \leq (1+\eps)w(\Q)
\end{eqnarray*}
Above we use the fact that $\Q$ is also a set-cover for $(\R', P')$. 
Scaling by $n/\epsilon w_{aprx}$, each set in $\R'$ has weight at least $1$, 
and weight of $\opt(\R', P') = O(n/\epsilon)$.
\end{proof}

Hence for the purpose of a $(1+\epsilon)$-approximation, we can assume that the minimum weight of any region is $1$ and the  weight of the optimal set-cover
is $O(n/\epsilon)$.

\section{QPTAS for Weighted Pseudodisks in $\Re^2$}
\label{sec:gsc}

Our main result in this section is:

\begin{theorem}
\label{lem:rajnikant}
Let $\mathcal{R}= \{R_1, \ldots, R_n\}$ be a set of $n$ weighted $\alpha$-simple pseudodisks with minimum weight $1$. Let $P$ be a set of points in the plane, with no point lying on the boundary of any of the pseudodisks. 
Assume also that  no pseudodisk in $\opt(\R, P)$ has weight more than $w(\opt(\R,P))/3$.
Then for any $\delta>0$, there exists a curve $\C$ such that 
\begin{itemize}
\item $w(\textsc{Opt}(\mathcal{R}, P_{in}(\C))) \leq (\frac{2}{3}+3\delta) w(\textsc{Opt}(\mathcal{R}, P))$
\item $w(\textsc{Opt}(\mathcal{R}, P_{ext}(\C))) \leq (\frac{2}{3}+3\delta)w(\textsc{Opt}(\mathcal{R}, P))$
\item $w(\textsc{Opt}(\mathcal{R}, P_{in}(\C))) + w(\textsc{Opt}(\mathcal{R}, P_{ext}(\C)))  \leq (1+2\delta)w(\textsc{Opt}(\mathcal{R}, P))$
\item The complexity of $\C$ is $O(\frac{\alpha}{\delta^2} \log w(\textsc{Opt}(\R,P))$.
\end{itemize}
\end{theorem}

A QPTAS for weighted pseudodisks follows
from this theorem in similar manner
to that of  Adamaszek-Wiese~\cite{AW13, AW14}.
We first use Claim~\ref{claim:tofame} to reduce the given instance of the set cover problem to an instance $(\R,P)$ where the minimum weight of the regions is $1$ and the weight of the optimal solution $w(\textsc{Opt}(\R,P))$ is $O(n/\epsilon)$.
Assume also that no pseudodisk in $\opt(\R, P)$ has weight more than $w(\opt(\R,P))/3$.
The input instance can be easily perturbed so that no point lies on the boundary of any region. Now by applying Theorem~\ref{lem:rajnikant}
with a given $\delta$, there exists a curve $\C$
of complexity $O( 1/\delta^2 \log(n/\epsilon))$. Thus by enumeration, 
there are $O(n^{(1/\delta^2 \cdot \log n/\eps)^{O(1)}})$ such possible curves $\C'$
(the proof of Theorem~\ref{lem:rajnikant} shows that the vertices
of any such $\C'$ come from a polynomial-sized subset
that can be computed in polynomial time)
each giving two sub-problems $(\R, P_{in}(\C'))$ and $(\R, P_{ext}(\C'))$.
Thus, as $\eps$ is a constant, the problem is QPT-partitionable, which together with Lemma~\ref{thm:alg}
gives the required QPTAS.
Finally, note that there can be at most $2$ pseudodisks in the optimal solution
with weight more than $w(\opt(\R,P))/3$, and one can simply guess (by
enumerating the at most $O(n^2)$ possibilities)
these pseudodisks, and recurse on the sub-problem where the weight
of the optimal solution is reduced by a constant-factor.



\begin{figure}
\vspace{-0.4in}
\hspace{-0.4in}
\begin{subfigure}[b]{0.23\textwidth}
\centering
  \includegraphics[width=7cm]{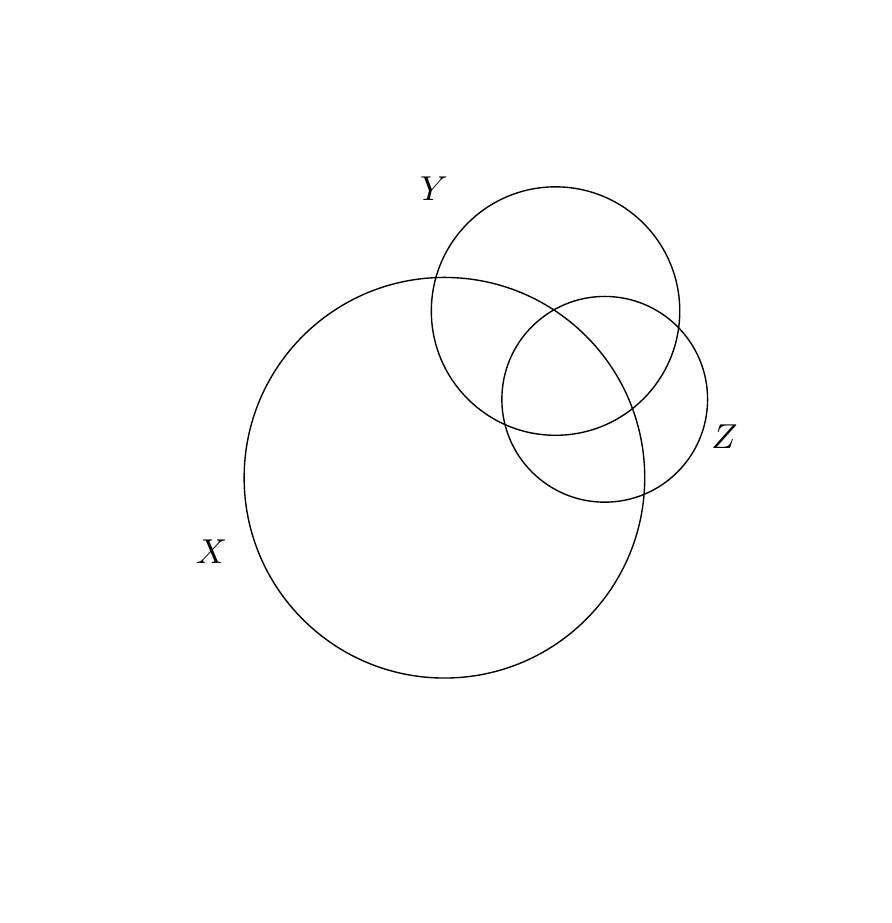} 
  \caption{$X$, $Y$, $Z$} 
  \label{fig:threedisks}
\end{subfigure}
\begin{subfigure}[b]{0.23\textwidth}
\centering
  \includegraphics[width=7cm]{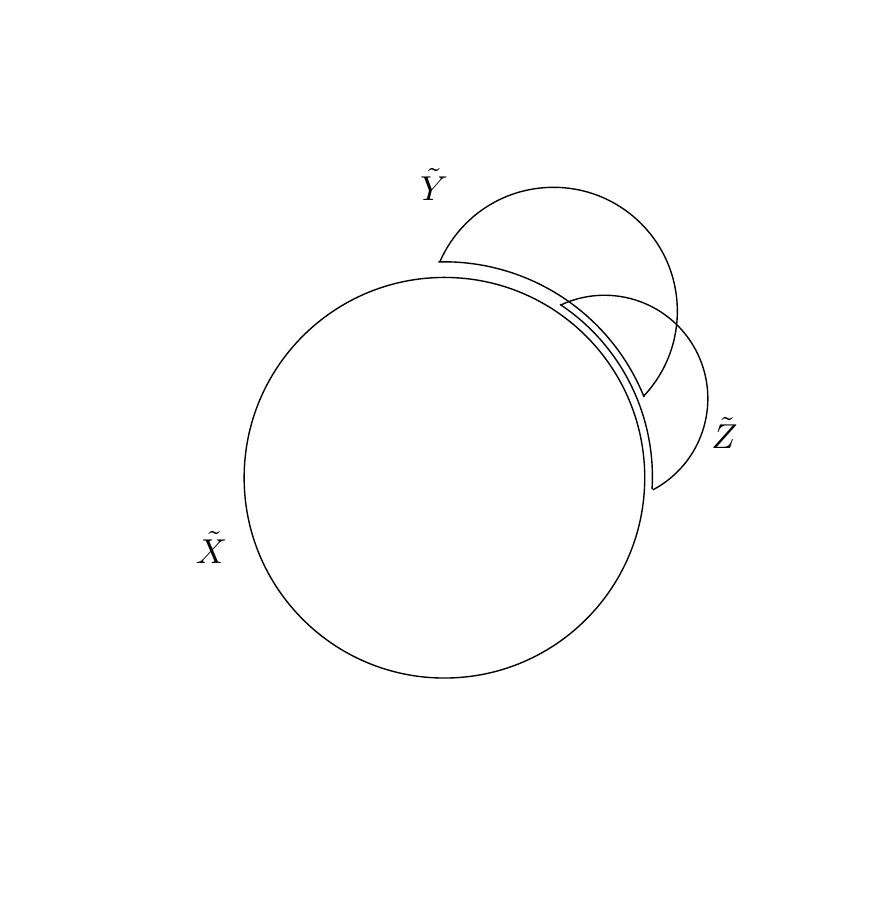} 
  \captionof{figure}{$\tilde{X}, \tilde{Y}$, $\tilde{Z}$}
  \label{fig:threedisks1}
\end{subfigure}
\begin{subfigure}[b]{0.23\textwidth}
\centering
  \includegraphics[width=7cm]{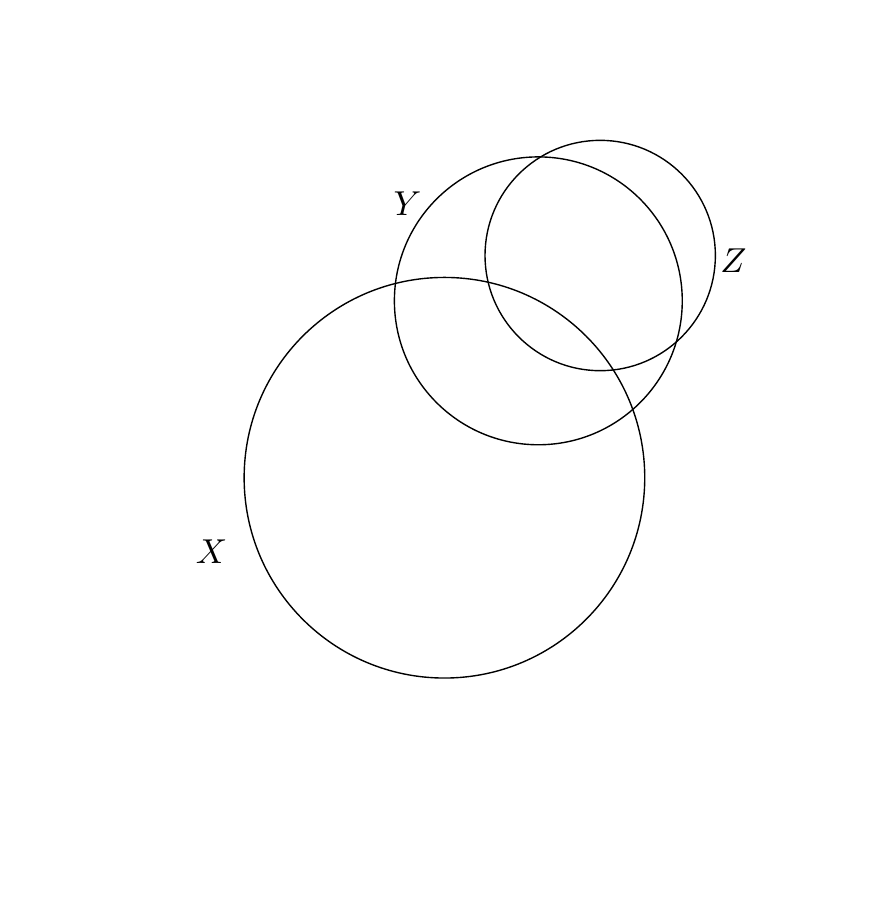} 
  \captionof{figure}{$X$, $Y$ and $Z$}
  \label{fig:threedisks2}
\end{subfigure}
\begin{subfigure}[b]{0.23\textwidth}
\centering
  \includegraphics[width=7cm]{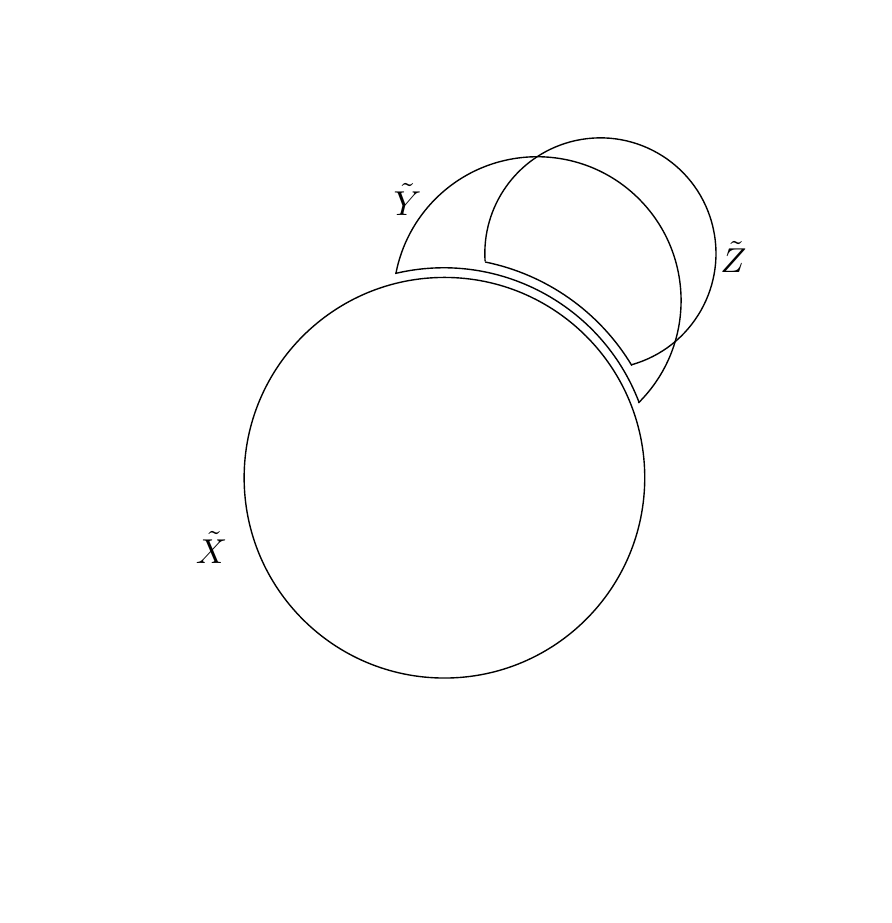}
  \captionof{figure}{$\tilde{X}, \tilde{Y}$, $\tilde{Z}$}
  \label{fig:threedisks3}
\end{subfigure}
\caption{Pushing pseudodisks}
\end{figure}

Towards proving the above theorem, we define structural decompositions
for pseudodisks in $\Re^2$ called
\emph{core decompositions}. We will also use this notion for halfspaces in $\Re^3$.
Informally, given a set of pseudodisks, our goal is to shrink them in such a way that their union remains (almost) unchanged but the number of vertices~\footnote{By vertices we mean the intersection points of the boundaries of the pseudodisks.} in the arrangement
decreases. 
\hide{
Suppose that we are given two pseudodisks as in Figure~\ref{fig:twodisks} and we want to remove both the intersections while fixing the union. We could keep $X$ as it is and replace $Y$ by $\tilde{Y} = closure(Y \setminus X)$ as shown in Figure~\ref{fig:twodisks1}. This keeps the union unchanged but $X$ and $\tilde{Y}$ now share a portion of their boundaries and thus technically do not form pseudodisks. To circumvent this problem, we replace our requirement that the union be unchanged with the requirement that any point that does not lie `too close' to the boundary of any of the pseudodisks is still covered. Then we can construct $\tilde{Y}$ as shown in Figure~\ref{fig:twodisks2}. Motivated by this, we have the definitions below. }
Denote by $B_\tau$ a closed ball of radius $\tau$. 
We denote the Minkowski sum by $\oplus$.

\begin{defn}({\bf Core Decomposition.})
Given $\R = \{R_1, \ldots, R_n\}$ and a $\beta>0$, a set of regions $\tilde{\R} = \{\tilde{R}_1, \ldots, \tilde{R}_n\}$
is called a \emph{$\beta$-core decomposition} of $\R$ (and each
$\tilde{R}_i$ a \emph{core} of $R_i$) if
(1) $\tilde{R}_i \subseteq R_i$ for all $i=1, \ldots, n$,
(2) $\bigcup_i \tilde{R}_i \supseteq \bigcup_i  {R}_i \setminus \bigcup_i(\partial R_i \oplus B_\beta) $, and
(3) each $\tilde{R}_i$ is simply connected.
\end{defn}

Each disk $\tilde{R}_i \in \tilde{\R}$ will be composed of pieces
of boundaries of the disks in $\R$. The sequence of the endpoints of these pieces
will be the vertices defining $\tilde{R}_i$ (denote this sequence
by $v(\tilde{R}_i)$, and its cardinality by $|v(\tilde{R}_i)|$).
The information needed to
uniquely determine $\tilde{R}_i$ then is the sequence of these pieces,
or equivalently, the sequence of vertices  defining $\tilde{R}_i$.
In the following, we will use the term {\em core decomposition} to mean a $\beta$-core decomposition with a suitably small $\beta > 0$ to be fixed later.
The following two lemmas show the existence of core decompositions
with specific properties. 

\begin{lemma}
\label{lemma:pusher}
Given a cover-free set $\R$ of pseudodisks, 
 a marked pseudodisk $X \in \R$ (called
the \emph{pusher}) and a $\beta > 0$, there exists a 
$\beta$-decomposition $\tilde{\R}$ of $\R$ such 
that $\tilde{X} = X$ and $\tilde{R}\cap \tilde{X} = \emptyset$
for all $R \neq X$ and $\{ \tilde{R}: R\in \R \}$ is a cover-free family of pseudodisks.
\end{lemma}
\begin{proof}
\hide{
In the case of two pseudodisks we can set $\tilde{Y}$ to be $closure(Y \setminus (X \oplus B_\mu))$ where $\mu$ is a small enough real less than $\beta$ chosen so that $\tilde{Y}$ remains a pseudodisk. In this case we say that $X$ pushes $Y$ with gap $\mu$. Now, suppose that we have three pseudodisks as shown in Figure~\ref{fig:threedisks}. We now want to push $Y$ and $Z$ out of $X$. However, if we push with the same gap $\mu$ then $\tilde{Y}$ and $\tilde{Z}$ would share a portion of their boundaries. We therefore need to push different gaps $\mu_Y$ and $\mu_Z$ to ensure that $\tilde{Y}$ and $\tilde{Z}$ are pseudodisks. In this case it does not matter whether $\mu_Y$ is smaller or $\mu_Z$ is smaller. However if the pseudodisks were as shown in Figure~\ref{fig:threedisks2}, then we must pick $\mu_Y$ to be smaller than $\mu_Z$ otherwise $\tilde{Y}$ and $\tilde{Z}$ would be piercing each other and hence not form a set of pseudodisks. Now we formally describe the way to push when there are an arbitrary number of pseudodisks.}

Set $\tilde{X} = X$. For each $R \in \R \setminus \{X\}$, we compute a number $gap(R) \in (0,\mu)$, where $\mu \in (0,\beta)$ is a suitably small number, and set $\tilde{R} = closure(R \setminus (X \oplus B_{gap(R)}))$ (we say that $X$ pushes $R$ with gap $gap(R)$).
For any $R \in \R \setminus \{X\}$, let $I_R$ be the interval $R \cap \partial X$ on the boundary of $X$.
As no pseudodisk in $\R$ is completely contained in any other
pseudodisk of $\R$ (cover-free), the intervals $I_R$ are well-defined. 
Consider the partial order $\preceq$ on these intervals
defined by inclusion ($I_R \preceq I_S$ if $I_R \subset I_S$). 
By a topological sorting of this partial order we can assign a distinct rank $rank(R) \in \{1,\ldots,n\}$ to each pseudodisk $R \in \R$ such that if $I_R \subset I_{R'}$ then $rank(R) > rank(R')$. 
We set $gap(R) = \mu \frac{rank(R)}{n}$.



Clearly each core is contained in its corresponding pseudodisk and 
for a small-enough $\mu$, it is simply connected. Also, since the points we may have removed from the union, due to the gaps, lie in $X \oplus B_\mu$, the second condition in the definition of  a $\beta$-core decomposition is satisfied. The cores obtained are also cover-free because the union of cores cover the union of the original regions (except close to boundaries). Since the input set is cover-free each pseudodisk has a {\em free} portion that is not covered by others. The core corresponding to a pseudodisk then must cover the free portion in that pseudodisk which is not covered by the other cores. Thus no core is covered by the union of other cores.

It remains to show that the cores form pseudodisks. Let $Y$ and $Z$ be any two pseudodisks in $\R$, and we now finish
the proof by showing
that the boundaries of $\tilde{Y}$ and $\tilde{Z}$ intersect at most twice. The possible cases are the following:
(1) $I_Y \cap I_Z = \emptyset$, 
(2) $I_Y \preceq I_Z$ and 
(3) $I_Y \cap I_Z \neq \emptyset$ and $I_Y \npreceq I_Z$.

In case $1$, since the intervals are disjoint, $\partial \tilde{Y}$ and $\partial \tilde{Z}$ do not have any new intersection that $\partial Y$ and $\partial Z$ did not have. They may have lost intersections lying in $X$. In any case, $\partial \tilde{Y}$ and $\partial \tilde{Z}$ intersect at most twice. In case $2$, $Y$ gets pushed with a smaller gap than $Z$ and the situation is exactly as shown in Figures ~\ref{fig:threedisks2} and ~\ref{fig:threedisks3}. In case $3$, $Y$ and $Z$ get pushed with different gaps and the situation is exactly as shown in Figures ~\ref{fig:threedisks} and ~\ref{fig:threedisks1}.
\end{proof}
\textbf{Remark:} Note that for each pseudodisk $R$ intersecting $X$, 
the boundary of $\tilde{R}$ now has two new vertices corresponding to the two intersections of $\partial R$ with $\partial X$. These 
vertices are slightly perturbed (and arbitrarily close) 
from the intersections because of 
the gap. We say that each such new vertex \emph{corresponds} to the
original intersection vertex between $\partial R$ and $\partial X$. 
When the context is clear, we will not distinguish between this new vertex
and the vertex it corresponds to.
\\

Before we prove our next main result on core decompositions,
we will need the following technical result. 
For clarity,  vertex $v$, the intersection point
of $R_i$ and $R_j$, is written as $(v,i,j)$.
Given $\R$, the \emph{depth} of a vertex $(v,i,j)$, denoted $d_v$,
is the total weight of the regions in $\R$ containing $(v,i,j)$ in the interior
(thus it excludes the weight of $R_i$ and $R_j$).

\begin{claim}
\label{claim:cs}
Let $\R = \{R_1, \ldots, R_n\}$ be a set of $n$ weighted pseudodisks,
and $k>0$ a given parameter. 
Assume $R_i$ has weight $w_i$, and $W = \sum_i w_i$. Further let 
$(v,i,j)$ denote a vertex $v$ in the arrangement of $\R$ defined by $R_i$
and $R_j$. Then
$$ \sum_{\substack{(v,i,j) \text{ s.t. } \\ k \leq d_v < 2k}} 
\frac{w_i \cdot w_j}{w_i+w_j+k} = O(W) $$
\end{claim}
\begin{proof}
The proof follows from melding the Clarkson-Shor technique
with a charging argument.
Let $\R_1 \subset \R$ be the set of disks with weight $w_i \geq 2k$, and $\R_2 = \R \setminus \R_1$.
Note that $|\R_1| = O(W/k)$, and that any vertex $v$ with depth less than $2k$
and defined by two disks in $\R_1$ must lie on the boundary of the union
of the regions in $\R_1$. This implies 
\begin{equation} 
\sum_{\substack{(v,i,j) \text{ s.t. } R_i,R_j \in \R_1 \\ k \leq d_v < 2k}} 
\frac{w_i \cdot w_j}{w_i+w_j+k} \leq \sum_{\substack{(v,i,j) \text{ s.t. } R_i,R_j \in \R_1 \\ k \leq d_v < 2k}} 
\min \{w_i, w_j \} =  O\left(\sum_{R_i \in \R_1} w_i \right) = O(W)
\end{equation}
where the second-last inequality follows from a charging
argument: the number of vertices in the union for $t$ pseudodisks is $O(t)$~\cite{PS09}, 
and so each vertex can be assigned to one of its two disks such that
each disk gets $O(1)$ assigned vertices. 

Set $S = \R_1$, and further add each disk $R_i \in \R_2$ into $S$
with probability $p_i = w_i/4k$. Then
$$ \text{Expected union complexity of S} = E[ O( |S| ) ] = O(|\R_1| + \sum_{R_i \in \R_2} \frac{w_i}{4k}) = O(W/k)$$
On the other hand, the expected number of vertices
defined by the intersection
of a disk in $\R_2$ and a disk in $\R_1$, and of depth less than $2k$,
that end up as vertices in the union of $S$:
\begin{eqnarray*}
\sum_{\substack{(v,i,j) \text{ s.t. } R_i \in \R_2, R_j \in \R_1 
\\ k \leq d_v < 2k}} \frac{w_i}{4k} \prod_{\substack{R_l \text{ contains } v
\\ R_l \in \R_2}}
(1-p_l) 
\geq \sum \frac{w_i}{4k} e^{-2\sum_l w_l/4k}  
\geq 
\sum \frac{w_i}{4ek}
\end{eqnarray*}
using the fact that $1-x \geq e^{-2x}$ for $0 \leq x \leq 0.5$,
and that $\sum w_l \leq 2k$ as all such $R_l$ contain $v$, which has
depth at most $2k$.
The expected number of vertices in the union of $S$ 
defined by two disks in $\R_2$, and of depth at most $2k$:
\begin{eqnarray*}
\sum_{\substack{(v,i,j) \text{ s.t. } R_i \in \R_2, R_j \in \R_2 
\\ k \leq d_v < 2k}} \frac{w_i}{4k} \cdot \frac{w_j}{4k} \prod_{\substack{R_l \text{ contains } v \\ R_l \in \R_2}} (1-p_l) 
\geq \sum \frac{w_i \cdot w_j}{16ek^2}
\end{eqnarray*}
Putting the lower- and upper-bounds together, we arrive at:
\begin{equation}
\label{eq:1}
\sum_{\substack{(v,i,j) \text{ s.t. } R_i \in \R_2, R_j \in \R_1}} \frac{w_i}{4ek} +
\sum_{\substack{(v,i,j) \text{ s.t. } R_i \in \R_2, R_j \in \R_2}} \frac{w_i \cdot w_j}{16ek^2} 
\leq O(W/k)
\end{equation}

Finally, 
\begin{eqnarray*}
\sum_{\substack{(v,i,j) \text{ s.t. } \\ k \leq d_v < 2k}} 
\frac{w_i \cdot w_j}{w_i+w_j+k} &=& 
\sum_{\substack{(v,i,j) \\ R_i \in \R_2, R_j \in \R_1}}
 \frac{w_i \cdot w_j}{w_i+w_j+k} + \sum_{\substack{(v,i,j) \\ R_i \in \R_2, R_j \in \R_2}}
 \frac{w_i \cdot w_j}{w_i+w_j+k}  \\ 
 && + \sum_{\substack{(v,i,j) \\ R_i \in \R_1, R_j \in \R_1}}
 \frac{w_i \cdot w_j}{w_i+w_j+k} \\
 & \leq & 
 \sum_{\substack{(v,i,j) \\ R_i \in \R_2, R_j \in \R_1}}
 w_i + \sum_{\substack{(v,i,j) \\ R_i \in \R_2, R_j \in \R_2}}
 \frac{w_i \cdot w_j}{k} + O(W) = O(W)
\end{eqnarray*}
where the last inequality follows from Equation~(\ref{eq:1}).

\end{proof}

\begin{figure}
\vspace{-0.4in}
\centering
\begin{minipage}{.3\textwidth}
  \centering
  \includegraphics[width=7cm]{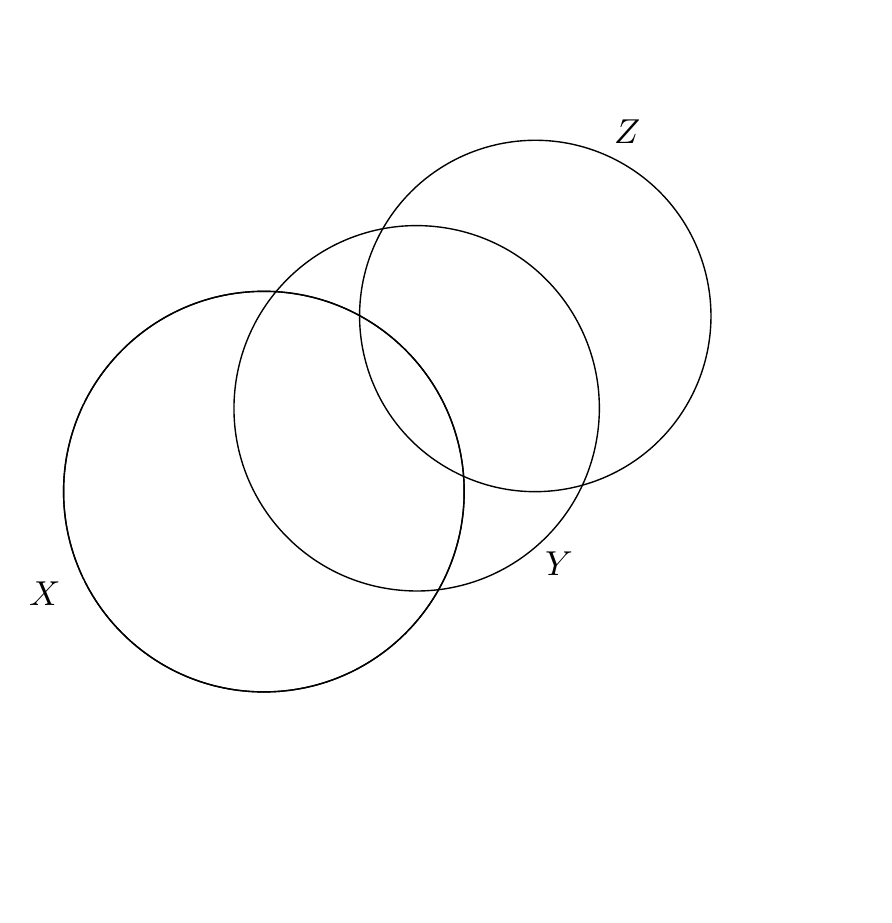}
\end{minipage}
\begin{minipage}{.3\textwidth}
  \centering
  \includegraphics[width=7cm]{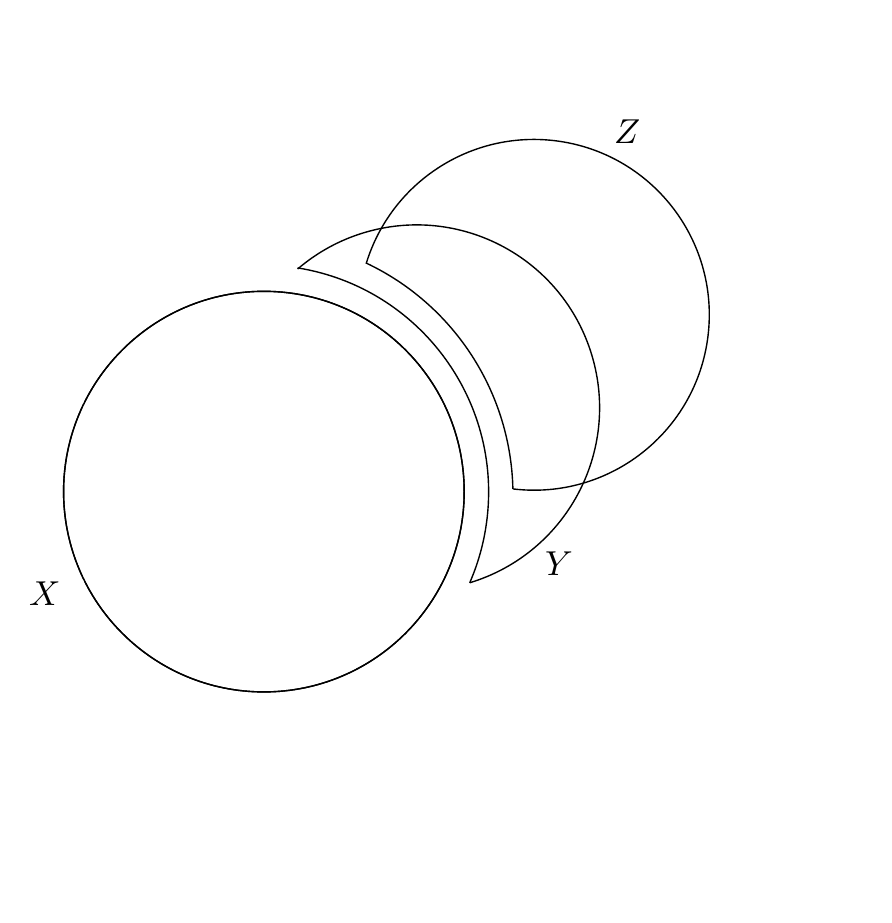}
\end{minipage}
\begin{minipage}{.3\textwidth}
  \centering
  \includegraphics[width=7cm]{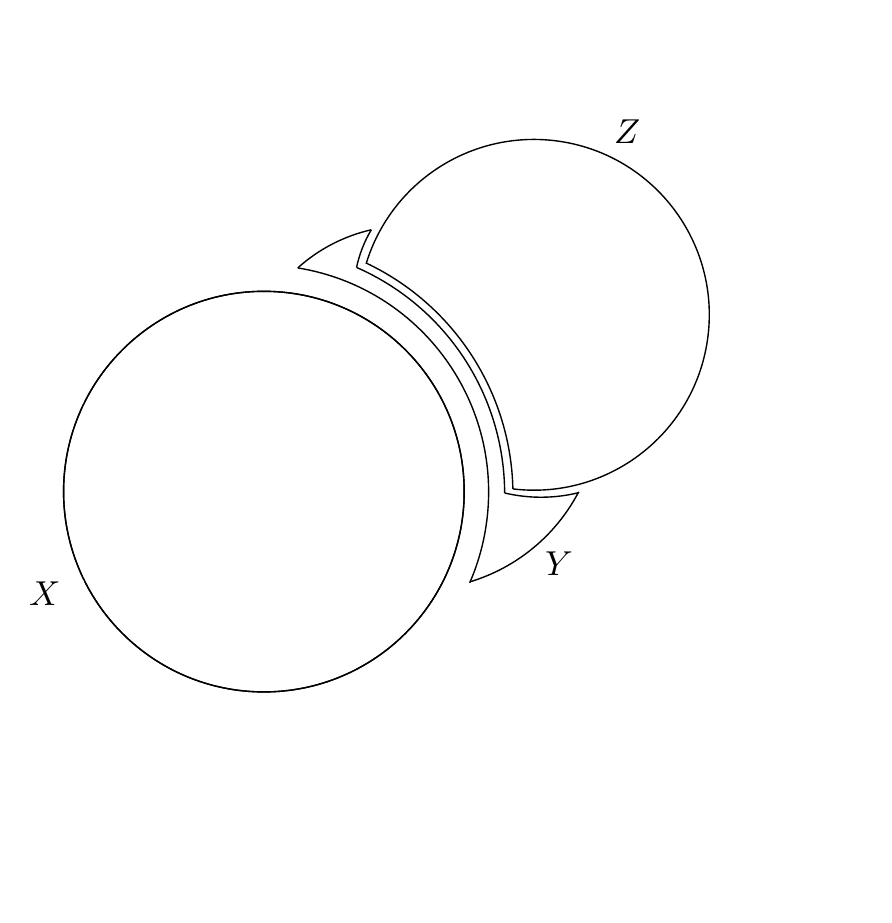}
\end{minipage}
\vspace{-0.7in}
\caption{$X$ is the first pusher, $Z$ is the second pusher and $Y$ is the third pusher.}
\label{fig:tunnels}
\end{figure}

\begin{lemma}
\label{lemma:disjointcoredecomposition}
Let $\R = \{R_1, \ldots, R_n\}$ be a set of $n$ weighted pseudodisks,
where $R_i$ has weight $w_i$, and $W = \sum_i w_i$. 
Then there exists
a core decomposition of $\R$, say the set $\tilde{\R} = \{\tilde{R}_1,
\ldots, \tilde{R}_n\}$, such that the pseudodisks in $\tilde{\R}$ are pairwise disjoint, and
$$\sum_{i} |v(\tilde{R}_i)| \cdot w_i = O(W \log W)$$
\end{lemma}
\begin{proof}
Recall that $v(\tilde{R}_i)$ is the sequence of vertices 
defining $\tilde{R}_i$ (which will be composed
of pieces of boundaries of regions in $\R$), and determines $\tilde{R}_i$.

The algorithm to construct $\tilde{\R}$ is the following.
Construct a permutation, say $\pi$, of $\R$ randomly w.r.t. to the weight distribution
of the disks as follows. Pick a random disk, where $R_i$ is picked with probability 
$w_i/W$. Set this disk to be the first disk in the permutation, and recursively
construct the rest of the permutation on the remaining disks. Let $\pi_i$
denote the position of $R_i$ in this permutation. 
Let $\R_0 = \R$. Apply Lemma~\ref{lemma:pusher} (with $R_{\pi^{-1}(1)} \in \R_0$ as the pusher) 
to get a core decomposition, denoted by $\R_1$,
of $\R_0$. Now apply Lemma~\ref{lemma:pusher} on $\R_1$
(with $\tilde{R}_{\pi^{-1}(2)} \in \R_1$ as the pusher) to get the set
$\R_2$. 
Continuing iteratively by applying
Lemma~\ref{lemma:pusher} with the successive core of each $R_{\pi^{-1}(i)}$ in
$\R_i$ as the pusher,
we get the set $\tilde{\R} = \R_n$. See Figure~\ref{fig:tunnels} for an example.

We have the following crucial fact:
\begin{claim}
If $v(\tilde{R}_l)$ has a vertex corresponding to $(v,i,j)$, 
 then $(v,i,j) \in R_l$ and furthermore,
\[ \left\{
  \begin{array}{l l}
    \max \{\pi_i, \pi_j \} < \min_{R_m \ni v} \pi_m & \quad \text{if } l \neq i,j \\
    \pi_j < \pi_i & \quad \text{if } l = i \\
    \pi_i < \pi_j & \quad \text{if } l=j
  \end{array} \right.\]
\end{claim}
\begin{proof}
First consider the case when $l \neq i,j$. 
The proof relies on the observation that by the proof of Lemma~\ref{lemma:pusher},
if $v$ is a vertex in the arrangement of $\R$, and at any point a region
containing $v$ in its interior is used as a pusher, then
any  $\tilde{R} \in \tilde{\R}$ cannot have a vertex that corresponds
to $v$. Thus the only way a vertex corresponding to $(v,i,j)$
can be part of the boundary of any $\tilde{R}$  is if both
the regions $R_i$ and $R_j$ occur earlier in $\pi$ than
any of the regions containing $v$.
For the case when $l = i$ (same for when $l=j$), the proof follows from the fact that 
if at any point there are two regions $X$ and $Y$ and $X$ is used as a pusher
before $Y$, then the core of $X$ cannot have a vertex that corresponds to 
a vertex defined by $X$ and $Y$ (see Figure~\ref{fig:tunnels}).
%
\end{proof}

Then
$$ \sum_{i}
|v(\tilde{R}_i)| \cdot w_i = 
\sum_{(v,i,j)} \sum_{l \text{ s.t. } v \in v(\tilde{R}_l)} w_l
= \sum_{(v,i,j)} \sum_{l \text{ s.t. } v \in R_l} X_{(v,i,j,l)} \cdot w_l $$ 
where the indicator variable $X_{(v,i,j,l)} =1 $ iff $v \in v(\tilde{R}_l)$ (more precisely,
the vertex corresponding to $v$ is in $v(\tilde{R}_l)$).
Using the above Claim and Claim~\ref{claim:cs}, we calculate the expected value of the required bound:

\begin{eqnarray*}
\sum_{i} E \big[ |v(\tilde{R}_i)|  \big] \cdot w_i & = & \sum_{(v,i,j)} \sum_{l \text{ s.t. } v \in R_l} E \big[ X_{(v,i,j,l)}  \big]\cdot w_l \\
& = & \sum_{(v,i,j)} \sum_{l \text{ s.t. } v \in R_l} \Pr \big[ v \in \tilde{R}_l  \big]\cdot w_l \\
& = & \sum_{(v,i,j)} \Bigl( \Pr[ v \in \tilde{R}_i]w_i + \Pr[ v \in \tilde{R}_j]w_j
+ \sum_{\substack{i,j \neq l \\ v \in R_l}}  \Pr[ v \in \tilde{R}_l]w_l \Bigr) \\
& = & \sum_{(v,i,j)} \Bigl( \frac{w_j}{w_i+w_j+d_v}\frac{w_i}{w_i+d_v}w_i  
+ \frac{w_i}{w_i+w_j+d_v}\frac{w_j}{w_j+d_v}w_j \\
& & 
+ \sum_{\substack{i,j \neq l \\ v \in R_l}}  (\frac{w_j}{w_i+w_j+d_v}\frac{w_i}{w_i+d_v}  + \frac{w_i}{w_i+w_j+d_v}\frac{w_j}{w_j+d_v})   w_l  \Bigr) \\
&=& \sum_{(v,i,j)} \Bigl(  
( \frac{w_j}{w_i+w_j+d_v}\frac{w_i}{w_i+d_v} ) \cdot (w_i + \sum_{R_l \ni v} w_l) \\
&& + (\frac{w_i}{w_i+w_j+d_v}\frac{w_j}{w_j+d_v}) \cdot (w_j + \sum_{R_l \ni v} w_l) \Bigr) \\
& = &  2 \sum_{(v,i,j)}  \frac{w_i \cdot w_j}{w_i+w_j+d_v}
 = 2 \sum_k \sum_{\substack{(v,i,j) \\ 2^k \leq d_v < 2^{k+1}}}  \frac{w_i \cdot w_j}{w_i+w_j+d_v} = O(W \log W).
\end{eqnarray*}
where the last inequality follows from Claim~\ref{claim:cs}.
\end{proof}

We can now finish the proof of Theorem~\ref{lem:rajnikant}.

\begin{proof}[Proof of Theorem~\ref{lem:rajnikant}]
We use
Lemma~\ref{lemma:disjointcoredecomposition} to obtain a core $\tilde{Q}$ for each $Q 
\in \Q = \textsc{Opt}(\R,P)$ and we assign to $\tilde{Q}$ the same weight as $Q$. 
Since the regions in $\Q$ cover $P$ their cores also cover $P$.~\footnote{Since no point lies on the boundary of any of the regions, there is a suitable choice of $\beta$ so that using $\beta$-core decompositions, we do not miss any of the points.}
As before, we denote the number of vertices in $\tilde{Q}$ by $|v(\tilde{Q})|$. 
By Lemma~\ref{lemma:disjointcoredecomposition}, $\sum_{Q \in \Q} |v(\tilde{Q})| w(Q) = O(w(\Q)
 \log{w(\Q)}))$. 
We set $\tau = C\cdot\frac{1}{\delta}\log{w(\Q)}$ 
for some large enough constant $C$.
Then by averaging,
$\sum_{Q \in \Q: |v(\tilde{Q})| >\tau } w(Q) < \delta w(\Q)$.
 
  Let $\Q_s = \{ Q \in \Q: |v(\tilde{Q})| \leq \tau \}$ and let $\tilde{\Q}_s = \{ \tilde{Q} : Q \in \Q_s \}$.
  The regions in $\tilde{\Q}_s$ are $\alpha\tau$-simple since they have at most $\tau$ sides and 
each of the sides is a portion of the boundary of a single $\alpha$-simple region in $\R$.
These regions have a total weight of $w(\Q_s)$. Thus applying Theorem~\ref{thm:weightedseparator}, we get separator $\C$ so that the total weight of the regions of $\Q_s$ whose cores lie in 
$interior(\C)$ ($exterior(\C)$) is at most $\frac{2}{3}w(\Q_s)$. Since the total weight of the regions in
  $\Q \setminus \Q_s$ is at most $\delta w(\Q)$, the total weight of all the cores that lie in
   $interior(\C)$ ($exterior(\C)$) is at most $(\frac{2}{3}+\delta)w(\Q)$. Also, the
    total weight of the cores in $\Q_s$ that intersect $\C$ is at most $\delta w(\Q_s)$. Thus the
     total weight of all the cores in $\Q$ that intersect $\C$ is at most $2\delta w(\Q)$. The complexity of $\C$ is $O(\alpha\tau/\delta) = O(\frac{\alpha}{\delta^2} \log{w(\Q)} )$, satisfying the fourth item in the statement of the theorem.

Let $\mathcal{Q}_1$ ($\mathcal{Q}_2$) be the set of regions whose cores are in $interior(\C)$ ($exterior(\C)$).
Let $\mathcal{Q}_3 = \mathcal{Q} \setminus \{\mathcal{Q}_1 \cup \mathcal{Q}_2\}$. 
Observe that the cores of the regions in $\mathcal{Q}_{1} \cup \mathcal{Q}_{3}$ cover all the points in $P_{in}(\C)$ and therefore the regions in $\mathcal{Q}_{1} \cup \mathcal{Q}_{3}$ themselves cover the points in $P_{in}(\C)$. Similarly the regions in $Q_{2} \cup \mathcal{Q}_{3} $ cover the points in $P_{ext}(\C)$.
 Therefore, 
$$w(\textsc{Opt}(\mathcal{R}, P_{in})) \leq w(\mathcal{Q}_{1} \cup \mathcal{Q}_{3}) = w(\mathcal{Q}_{1}) + w(\mathcal{Q}_{3}) \leq (\frac{2}{3}+3\delta)w(\mathcal{Q})$$
This proves the first item in the statement of the theorem. The second item is proved analogously. For the third item, we combine the inequalities $w(\textsc{Opt}(\mathcal{R}, P_{in})) \leq w(\mathcal{Q}_{1}) + w(\mathcal{Q}_{3})$ and $w(\textsc{Opt}(\mathcal{R},P_{ext})) \leq w(\mathcal{Q}_{2}) + w(\mathcal{Q}_{3})$. We get 
$$w(\textsc{Opt}(\mathcal{R}, P_{in})) + w(\textsc{Opt}(\mathcal{R}, P_{ext})) \leq w(\mathcal{Q}_{1}) + w(\mathcal{Q}_{2}) + 2w(\mathcal{Q}_{3}) \leq w(\mathcal{Q}) + w(\mathcal{Q}_{3}) \leq (1+2\delta)w(\mathcal{Q})$$ 
That proves the third item.

\end{proof}


{\bf Remark:} The above QPTAS can be extended to work for more general regions called {\em non-piercing regions} or {$r$-admissible regions}. For this only Lemma~\ref{lemma:pusher} needs to be extended to work for these regions. In this case, a region $R$ may intersect the boundary of pusher $X$ in more than one interval. To ensure that after pushing the new regions are still non-piercing, different gaps are required in different intervals for the same region $R$. This makes it technically more complicated. The details are given in Appendix~\ref{appendix:r-admissible}.

\section{QPTAS for Weighted Halfspaces in $\Re^3$}
Let $\H = \{H_1, \cdots, H_n\}$ be a set of halfspaces in $\Re^3$ where the halfspace $H_i$ has weight $w_i \geq 0$ and $W$ total weight. Let  $P$ be a set of points in $\Re^3$. Given $\H$ and $P$, we show that the problem of computing a subset of $\H$ of minimum weight whose union covers $P$ is QPT-partitionable, and then
Lemma~\ref{thm:alg} implies the QPTAS.

Consider the optimal solution $\opt = \opt(\H, P)$ for the problem, and let 
$W$ be the total weight of the halfspaces in $\opt$. For any halfspace $H$, define $\overline{H}$ to be other halfspace defined by its boundary $\partial H$ i.e., $\overline{H} = closure (\Re^3 \setminus H)$. For any set of halfspaces $\S$,
define  $\overline{\S} = \{\overline{H}: H \in \S \}$. 
\begin{lemma}
If $\bigcup_{H \in \opt} H = \Re^3$, then one can compute $\opt(\H, P)$ in polynomial time.
\end{lemma}
\begin{proof}
If $\bigcup_{H \in \opt} H = \Re^3$, then by definition $\bigcap_{H \in \opt} \overline{H} = \emptyset$.
By Helly's theorem~\cite{M02} applied to the set of convex
regions in $\overline{\opt}$, it follows that then there must
be a subset $\opt' \subset \opt$ of at most $4$ halfspaces
such that $\bigcap_{H \in \opt'} \overline{H} = \emptyset$. In other words,
$\opt'$ covers $\Re^3$. As $\opt$ was a minimal-weight set cover,
it follows that $|\opt| \leq 4$. By enumerating all $4$-tuples of halfspaces in $\H$,
one can compute the optimal set-cover in polynomial time.
\end{proof}
From now on we assume that there is a point $o$ that does not lie in any of the halfspaces in $\opt$ (say the origin).
We will also assume without loss of generality that the intersection of halfspaces in $\overline{\opt}$  is a bounded polytope. This can be easily done by adding to the input four halfspaces with weight $0$ which do not contain any of the points in $P$ whose complements intersect in a bounded simplex. These four dummy halfspaces can then be included in any optimal solution without affecting the weight of the solution. 
Note also that each halfspace $H \in \opt$ must be part of some facet (in
fact, a unique facet)
of this polytope; otherwise $H$ is contained in the union of $\opt \setminus \{H\}$,
contradicting the set-cover minimiality of $\opt$.

We now define a core decomposition for the halfspaces in $\opt$ that allows a cheap balanced polyhedral separator. 

Consider the set system in which the base set are the halfspaces in $\opt$ and subsets are defined by taking any segment $ox$ with one end-point at $o$ and taking the set of halfspaces whose boundaries intersect the segment. More formally, for any $x \in \Re^3$, let $R_x = \{H \in \opt : \partial H \cap ox \neq \emptyset \} $. We now define $\R$ as the set $\{R_x : x\in \Re^3 \}$. Consider the weighted set 
system $(\H, \R)$, where the weight of $w(R)$ of any $R \in \R$ is the sum of the weights of the halfspaces in $R$. 

\begin{lemma}
The VC-dimension of $(\H, \R)$ is at most $3$.
\end{lemma}
\begin{proof}
For two distinct points $x$ and $y$ that lie in the same cell of the arrangement of the halfspaces in $\opt$, $R_x = R_y$. 
So for any subset of $\opt$ of size $k$, the number of induced subsets is at most the number of cells in an arrangement of these $k$ halfspaces, which is at most
 $\binom{k}{3}+\binom{k}{2}+\binom{k}{1}+\binom{k}{0}$. For $k = 4$, this number is less than $2^k$, implying that no subset of size $4$ is shattered. Thus the VC-dimension this set system is at most $3$. 
\end{proof}
Thus, by the $\epsilon$-net theorem~\cite{HW87}, there is an $\epsilon$-net for this set system of size $O(\frac{1}{\epsilon}\log{\frac{1}{\epsilon}})$. 
Let  $N$ be an $\epsilon$-net for this set system for a value of $\epsilon$ to be
fixed later. As before, we will assume that the intersection of halfspaces in $\overline{N}$ is a bounded polytope $\P$. This can be ensured by including in $N$ the dummy halfspaces.

For any set $S \subset \Re^3$, define $cone(S)$ to be the set 
$\{\lambda x: x \in S, \lambda \geq 0 \}$. For any halfspace $H \in N$, we define the core of $H$ to be $\tilde{H} = cone(f) \cap H$ where $f$ is the facet of $\P$ corresponding to $H$ i.e., the facet contained in $\partial H$. Note that each halfspace in $N$ (with the exception of dummy halfspaces) has a unique facet of $\P$ corresponding to it. 
 For any halfspace $H \in \opt \setminus N$, we defined the core as $\tilde{H} = \P\cap H$. The core of each halfspace is clearly contained in the halfspace and the union of these cores is clearly the same as the union of the halfspaces in $\opt$. 

We now assign a weight to each of the facets of $\P$ by distributing the weights of the halfspaces in $\opt$ to the facets so that the total weight of the faces is the same as the total weight of the halfspaces. The weight of each halfspace in $N$ is assigned to facet corresponding it. For a halfspace in $H \in \opt\setminus N$, we distribute its weight equally among all faces $f$ s.t. $cone(f)$ intersects the core $\tilde{H}$ of $H$. 

The $1$-skeleton of $\P$ is a planar graph $G$ and we have assigned weights to its faces. Let $n'$ denote the number of vertices in this graph; note that
$n' = O(|N|) = O(\frac{1}{\epsilon}\log{\frac{1}{\epsilon}})$. 
By \cite{M86}, there exists in this graph a cycle separator $\C$ of size $O(\sqrt{n'})$  so that the total weight of the faces in the interior (exterior) of $\C$ is at most two thirds of the total weight. We show that the polytope $\hat{\C} = cone(\C)$ is the desired cheap balanced separator for the cores we have defined. $\hat{\C}$ splits $\Re^3$ into two connected pieces whose closures we call the interior and the exterior of $\hat{\C}$. The choice is arbitrary. Note that $interior(\hat{\C}) \cap exterior(\hat{\C}) = \hat{C}$.

First note that for each core that lies in the interior (exterior), the weight of the corresponding halfspaces is distributed only among the faces of $\P$ lying in the interior (exterior) of $\hat{\C}$. Hence the total weight of all cores that lie in the interior (exterior) of $\hat{C}$ is at most two thirds of the total weight of all halfspaces. 

We now need to bound the total weight of the cores that cross $\hat{\C}$. None of the cores of the halfspaces in $N$ cross $\hat{C}$. Consider a halfspace $H \in \opt \setminus N$. Its core is defined as $H \cap \P$. If this core intersects $\hat{\C}$, then $H \cap \partial \P$ intersects $\C$. It follows that $H$ intersects an edge of $\C$ and thus must contain a vertex $v$ of $\C$. In other words, $\partial H$ intersects the segment $ov$. However, since $N$ is an $\epsilon$-net, for any vertex $v$, the total weight of halfspaces whose boundaries intersects $ov$ is at most $\epsilon W$. Since $\C$ has $O(\sqrt{n'})$ vertices, the total weight of all cores crossing $\hat{C}$ is $O(\sqrt{n'}\cdot \epsilon W)$. We set $\epsilon =A \delta^2 / \log {\delta^{-2}}$ for a suitable constant $A$ so that $\C$ has $O(\frac{1}{\delta}\log{\frac{1}{\delta}})$ vertices and the total weight of cores intersecting is at most $\delta W$. 

Finally, observe that the complexity of $\hat{\C}$ is determined by
the complexity of $\C$, and the point $o$. The vertices of
$\C$ are determined by intersections of $3$ halfspaces of $\H$, and so there
are $O(n^3)$ choices for each vertex of $\C$. To guess the point $o$,
it suffices to guess the cell of the arrangement of $\H$ in which it lies (there
are $O(n^3)$ such choices),
and pick any point in that cell. 

{\bf Remark:} It may appear that the set cover problem for halfspaces may 
be reduced to the problem for pseudodisks using techniques used in~\cite{MSW90}.
Unfortunately, that does not work because (i) we are in the weighted setting and (ii) because we cannot tolerate the loss of a constant factor when looking for a
$(1+\epsilon)$-approximation algorithm. It is also tempting to think that the technique used for halfspaces may be used for pseudodisks in the plane. That would mean taking an $\epsilon$-net $N$ for a suitable range space and then defining the core for each pseudodisks $R\notin N$ by removing from $R$ the portion of it covered by the union of pseudodisks in $N$. However, the problem in doing this is that the resulting cores may not be connected. This causes problems because if the cores are not connected then the cores not intersecting a separator curve $\C$ may still cover points in both $interior(\C)$ and $exterior(\C)$. 

\section{Conclusion}
In this paper we demonstrated the versatility of separator-based
algorithmic design on a problem seemingly unrelated
to the packing problems for which the separator had previously
been successfully applied. Getting a polynomial-time approximation 
scheme for the set-cover problem for weighted pseudodisks 
in the plane and weighted halfspaces in $\Re^3$
remains a very interesting open problem.

\newpage

\bibliographystyle{plain}
\bibliography{separators}

\newpage
\appendix

\section{QPTAS for Uniform Pseudodisks}
\label{appendix:uniformpseudodisks}

For the uniform pseudodisk case, there is an easier proof, which
we present now.

\begin{lemma}
\label{lemma:coredecomposition}
Given a set of pseudodisks $\R = \{R_1, \ldots, R_n\}$
and a parameter $\eta > 0$, there exists
a core decomposition of $\R$, say the set $\tilde{\R} = \{\tilde{R}_1,
\ldots, \tilde{R}_n\}$, such that 
\begin{enumerate}
\item the number of intersections in the arrangement
induced by $\tilde{\R}$ is $O(1/\eta^2 + \eta n^2)$, and
\item the number of vertices of each pseudodisk in $\tilde{\R}$ is $O(1/\eta^2)$.
\end{enumerate}
\end{lemma}
\begin{proof}
Construct an $\eta$-net, say the set $Q = \{Q_0, \ldots,
Q_t\}$, for $\R$. By
the result of Clarkson-Varadarajan~\cite{CV05}, $Q$ 
has size $O(1/\eta)$.

Let $\R_0 = \R$. Apply Lemma~\ref{lemma:pusher} (with $Q_0 \in \R_0$ as the pusher) 
to get a core decomposition, denoted by $\R_1$,
of $\R_0$. Now apply Lemma~\ref{lemma:pusher} on $\R_1$
(with $\tilde{Q}_1 \in \R_1$ as the pusher) to get the set
$\R_2$. 
Continuing iteratively by applying
Lemma~\ref{lemma:pusher} with the successive core of each $Q_i$ in
$\R_i$ as the pusher,
we get the set $\R_t$. Replace all 
the cores of pseudodisks $R \in Q$ in $\R_t$ by $R$ to get the set $\tilde{\R}$.

Observe that each $\tilde{R} \in \tilde{\R}$ 
is disjoint from each object in $Q$ (it became
disjoint from $Q_i$ latest at the $i$-th iteration). 
As $Q$ was an $\eta$-net, any point
in the plane not covered by the union of $Q$ has
depth at most $\eta n$.
Recall that by the Clarkson-Shor technique, as pseudodisks have linear union complexity,
the maximum number of vertices at depth at most $k$ is $O(nk)$.
Therefore the total number of intersections in $\tilde{\R}$
is 
$$O(|Q|^2 + |\R_t| \cdot \eta n) = O(1/\eta^2 + n \cdot \eta n)$$

This proves condition $1$.

For condition $2.$, from the proof of Lemma~\ref{lemma:pusher},
each boundary vertex of any core object corresponds to a vertex of the arrangement
induced by the objects in $Q$. As every pair of pseudodisks
can intersect at most twice, there are $O(1/\eta^2)$ vertices
in the arrangement of $Q$.
\end{proof}

We now give a proof of the existence of the separator for uniformly weighted pseudodisks.

\begin{lemma}\label{lem:dayofthejackal}
Given a set of $\mathcal{R}$ of $n$ uniformly weighted $\alpha$-simple pseudodisks (each with weight $1$) and set $P$ of points in the plane, no point lying on the boundary of any of the regions, and any parameter $\delta$, there exists a curve $\C$ such that 
\begin{itemize}
\item $w(\textsc{Opt}(\mathcal{R}, P_{in}(\C))) \leq (\frac{2}{3}+\delta) w(\textsc{Opt}(\mathcal{R}, P))$
\item $w(\textsc{Opt}(\mathcal{R}, P_{ext}(\C))) \leq (\frac{2}{3}+\delta)w(\textsc{Opt}(\mathcal{R}, P))$
\item $w(\textsc{Opt}(\mathcal{R}, P_{in}(\C))) + w(\textsc{Opt}(\mathcal{R}, P_{ext}(\C)))  \leq (1+\delta)w(\textsc{Opt}(\mathcal{R}, P))$
\item the complexity of $\C$ is  $O(\log^5 n/\delta^5)$ 
\end{itemize}
where $w(\mathcal{S})$ denotes the total weight of the regions in $\mathcal{S}$.
\end{lemma}
\begin{proof}
Consider the set $\mathcal{Q} = \textsc{Opt}(\mathcal{R}, P)$. We apply Lemma~\ref{lemma:coredecomposition} to $\mathcal{Q}$ and get
a core decomposition with 
 a core $\tilde{Q}$ for each $Q \in \mathcal{Q}$. Since, the regions in $Q$ cover $P$, their cores also cover $P$.\footnote{Since no point lies on the boundary of any of the regions, there is suitable choice of $\beta$ so that using $\beta$-core decompositions, we do not miss any of the points.} 
By the property of the core decomposition of Lemma~\ref{lemma:coredecomposition},
we have the total number of intersections to be $m = O(\delta^2 n^2 /\log^2 n)$,
and each core has $O(\log^4 n/\delta^4)$ vertices.
Since each of the curves forming the boundary of any core is composed of at most $\alpha$ $x$-monotone curves, the cores are $\alpha_1$-simple for $\alpha_1 = O(\alpha \log^4 n/\delta^4)$.
Applying Theorem~\ref{thm:separator} to these cores with parameter 
$r = \alpha^2 \log^{10} n / \delta^{10}$, we get a curve $\C$, where
\begin{eqnarray*}
\text{Number of cores intersected by }  \C &=& O( \sqrt{m+\alpha_1^2 n^2/r} )
= O(  \sqrt{ \frac{\delta^2 n^2}{\log^2 n} } )  = O( \frac{\delta n}{\log n} ) \\
\text{Complexity of } \C &=& O( \sqrt{r +\frac{mr^2}{\alpha_1^2 n^2}} ) 
= O( \sqrt{\frac{\log^{10}n}{\delta^{10}}} ) = O(\frac{\alpha \log^5 n}{\delta^5})
\end{eqnarray*}

Let $\mathcal{Q}_1$ ($\mathcal{Q}_2$) be the set of regions whose cores are in $interior(\C)$ ($exterior(\C)$).
Let $\mathcal{Q}_3 = \mathcal{Q} \setminus \{\mathcal{Q}_1 \cup \mathcal{Q}_2\}$. 
Observe that the cores of the regions in $\mathcal{Q}_{1} \cup \mathcal{Q}_{3}$ cover all the points in $P_{in}$ and therefore the regions in $\mathcal{Q}_{1} \cup \mathcal{Q}_{3}$ themselves cover the points in $P_{in}$. Similarly the regions in $Q_{2} \cup \mathcal{Q}_{3} $ cover the points in $P_{ext}$.
Theorem~\ref{thm:weightedseparator} guarantees that $w(\mathcal{Q}_{1})$ and $w(\mathcal{Q}_{2})$ are at most $\frac{2}{3}w(\mathcal{Q})$ and $w(\mathcal{Q}_{3}) \leq \delta w(\mathcal{Q})$. Therefore, 
$$w(\textsc{Opt}(\mathcal{R}, P_{in})) \leq w(\mathcal{Q}_{1} \cup \mathcal{Q}_{3}) = w(\mathcal{Q}_{1}) + w(\mathcal{Q}_{3}) \leq (\frac{2}{3}+\delta)w(\mathcal{Q})$$
This proves the first item in the statement of the theorem. The second item is proved analogously. For the third item, we combine the inequalities $w(\textsc{Opt}(\mathcal{R}, P_{in})) \leq w(\mathcal{Q}_{1}) + w(\mathcal{Q}_{3})$ and $w(\textsc{Opt}(\mathcal{R},P_{ext})) \leq w(\mathcal{Q}_{2}) + w(\mathcal{Q}_{3})$. We get 
$$w(\textsc{Opt}(\mathcal{R}, P_{in})) + w(\textsc{Opt}(\mathcal{R}, P_{ext})) \leq w(\mathcal{Q}_{1}) + w(\mathcal{Q}_{2}) + 2w(\mathcal{Q}_{3}) \leq w(\mathcal{Q}) + w(\mathcal{Q}_{3}) \leq (1+\delta)w(\mathcal{Q})$$ 
That proves the third item. 
Lemma~\ref{lemma:coredecomposition} also gives a set $\Gamma(\tilde{Q})$ for each $Q\in \mathcal{Q}$. Let $Y = \cup_{Q\in \mathcal{Q}} \Gamma(\tilde{Q})$. The separator $\C$ given by Theorem~\ref{thm:separator} can be described by a sequence of $O(\alpha\log^5 n/\delta^5)$ curves in $Y$ and $O(\alpha\log^5 n/\delta^5)$ additional bits. Since, by Lemma~\ref{lemma:coredecomposition} each curve in $Y$ can be described by a sequence of at most three curves in $\Gamma(\mathcal{R})$ and a constant number of bits, $\C$ can also be described by a sequence of  $O(\alpha\log^5 n/\delta^5)$ curves in $\Gamma(\mathcal{R})$ and $O(\alpha\log^5 n/\delta^5)$ additional bits of information.
\end{proof}

Now the existence of $\C$ immediately implies that the problem
is QPT-partitionable, which together with Lemma~\ref{thm:alg}
yields a QPTAS for the case with uniform weights.

\newpage
\section{QPTAS for $r$-admissible regions}
\label{appendix:r-admissible}
Two regions $A$ and $B$ are said to be non-piercing if they are simply connected, their boundaries intersect at most a finite number of times, and the regions $A \setminus B$ and $B \setminus A$ are connected. A finite set of regions is said to be non-piercing if they are pairwise non-piercing. Figure~\ref{fig:threeobjects} shows three regions that form a  non-piercing set. The intersection of any two regions $A$ and $B$ that are non-piercing consists of a disjoint union of {\em lenses} formed by the regions. Each lens is a connected component of $A\cap B$ which lies between two intersection points of $\partial A$ and $\partial B$ that are consecutive along both the boundaries. The shaded areas in Figure~\ref{fig:lenses} show the lenses in the intersection of a non-piercing pair of regions. Figure~\ref{fig:lens-cutting} shows a modification of the region $B$ is Figure~\ref{fig:lenses} so that one of the lenses is bypassed. The boundary of $B$ is modified so that the portion of $\partial B$ that lies inside $A$ is replaced by the dashed curve shown in the figure that is arbitrarily close to $\partial A$ but outside it. We will use such operations of {\em bypassing a lens} in the proof of the next lemma.

\begin{lemma}\label{lemma:pusher2}
Given a set $\R$ of non-piercing regions, 
and a marked region $X \in \R$ (called
the \emph{pusher}), there exists a 
core decomposition $\tilde{\R}$ of $\R$ such 
that i) $\tilde{X} = X$ and $\tilde{R}\cap \tilde{X} = \emptyset$ for all $R \neq X$ and ii)  $\tilde{\R}$ forms a family of non-piercing regions.
\end{lemma}
\begin{proof}
The idea for the proof is the same as with pseudodisks. We set $\tilde{X}=X$ and $X' = X \oplus B_\mu$ for a suitably small $\mu$. For sufficiently small $\mu$, $X'$ forms a non-piercing family with the regions in  $\R \setminus \{X\}$.

The intersection of each region in  $\R \setminus \{X\}$ forms a set of intervals on the boundary of $X$. Let $\I$ be the set of all these intervals corresponding to all the regions. We then consider the partial order among them defined by inclusion, just a we did in the case of pseudodisks. This time, instead of assigning ranks to the pseudodisks, we assign distinct ranks in the range $1$ to $m$ to the intervals, where $m$ is the total number of intervals in $\I$. Different intervals corresponding to the same region are given different ranks by this procedure.

Unlike with pseudodisks, we do not push a region $R$ with the single gap. Instead we push it with different gaps along the different intervals in which it intersects $\partial X$. The gap along an interval $I$ is $gap(I) = \mu \frac{rank(I)}{m}$. 
To better understand the pushing procedure, we imagine pushing along the intervals in $\I$ one by one in decreasing order of their ranks.

\begin{figure}
\centering
\begin{minipage}{.4\textwidth}
  \centering
  \includegraphics[width=7cm]{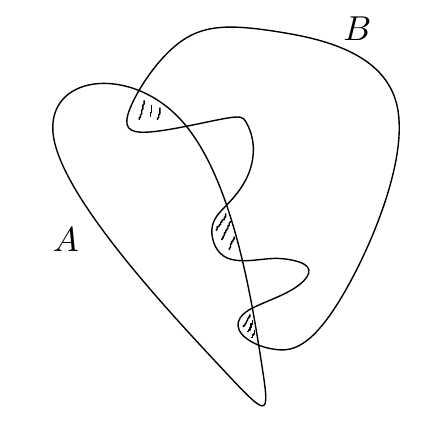}
  \captionof{figure}{Lenses in the intersection\\
   of two non-piercing regions.}
  \label{fig:lenses}
\end{minipage}
\begin{minipage}{.5\textwidth}
  \centering
  \includegraphics[width=7cm]{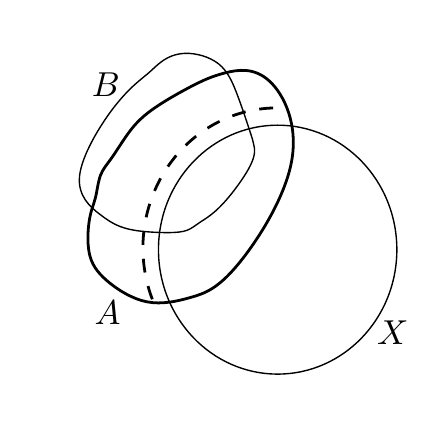}
  \captionof{figure}{The algorithm pushes $B$ before $A$.\\ Therefore the situation shown in the figure\\ cannot happen.}
  \label{fig:forbidden-cutting}
\end{minipage}%
\end{figure}

\begin{figure}
\centering
\includegraphics[width=7cm]{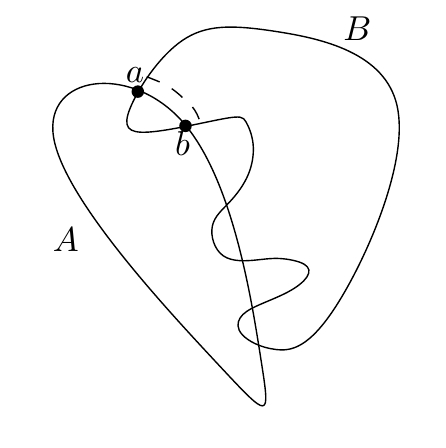}
\caption{$B$ is modified to bypass a lens.}
  \label{fig:lens-cutting}
\end{figure}

Figures~\ref{fig:retraction-to-boundary} and ~\ref{fig:retraction-to-boundary2} show the effect of pushing the region $A$ along the interval $I = [a,b]$ which is one of two intervals on $\partial X$ in which $A$ intersects $\partial X$. The part of $\partial A$ inside $X$ joining $a$ and $b$ is replaced by the dashed curve shown in Figure~\ref{fig:retraction-to-boundary2} that lies on the boundary of $X \oplus B_{gap(I)}$. 

\begin{figure}
\begin{minipage}{.5\textwidth}
  \centering
  \includegraphics[width=7cm]{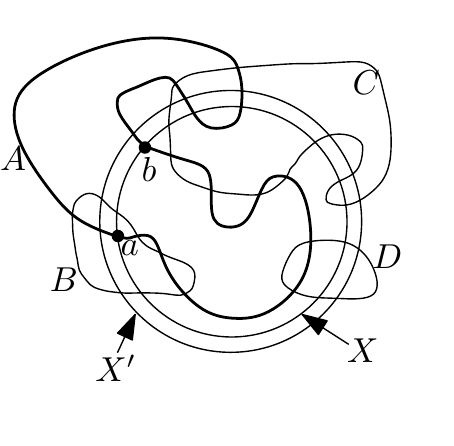}
  \captionof{figure}{Before $X$ pushes $A$ and $B$}
  \label{fig:retraction-to-boundary}
\end{minipage}%
\begin{minipage}{.5\textwidth}
  \centering
  \includegraphics[width=7cm]{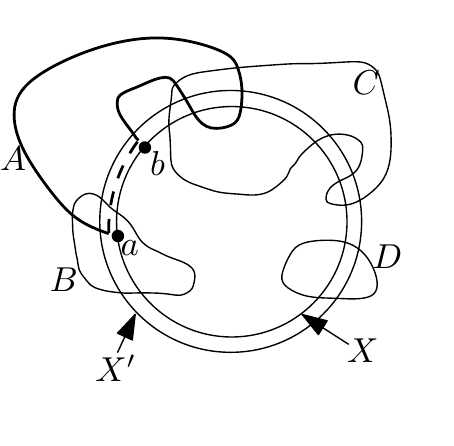}
  \captionof{figure}{After $X$ pushes $A$ and $B$}
  \label{fig:retraction-to-boundary2}
\end{minipage}
\end{figure}

We argue that this change to a region $A$ keeps the set of regions non-piercing. That is, the modified $A$ forms a non-piercing pair with every other object $B$ (which remains unchanged). To see this, imagine the change in $A$ as the net result of a sequence of changes. 
Within $X$, $A$ may form $0$ or more lenses with $B$. We first make $A$ bypass each of these lenses one by one.
 As discussed before, each of these changes keeps $A$ and $B$ non-piercing.
  After this, the number of intersections between the boundaries of $A$ and $B$ within $X$ is at most $1$. At this point we may move the boundary of $A$ to its final position on the boundary of $X \oplus B_{gap(I)}$. Observe that this final step does not change the number of intersections between $\partial A$ and $\partial B$ unless $B$ intersects $\partial X$ in a sub-interval  $I'\subset I$, as shown in Figure~\ref{fig:forbidden-cutting}. However, in that case, we push $B$ along the interval $I'$ with a larger gap, before pushing $A$ along $I$ since $I'$ gets a larger rank than $I$. 
 Thus, the modified $A$ bypasses some of the lenses with other objects (one with $B$, two with $C$ and one with $D$ in Figure~\ref{fig:retraction-to-boundary}) but this does not affect the non-piercing property of the family of regions. As another example,
  Figures~\ref{fig:threeobjects} and ~\ref{fig:threeobjects2} show the effect of pushing two regions $A$ and $B$ with a third one $X$. 

\end{proof}

\textbf{Remark:} Note that, as in the case of pseudodisks, for each $R \in \R$, the boundary of $\tilde{R}$ gains as many vertices as the number of intersections between $\partial X$ and $\partial R$. As before, these vertices correspond to the intersections but are slightly perturbed from the intersections because of the pushing with non-zero gaps.\\

\begin{figure}
\begin{minipage}{.5\textwidth}
  \centering
  \includegraphics[width=7cm]{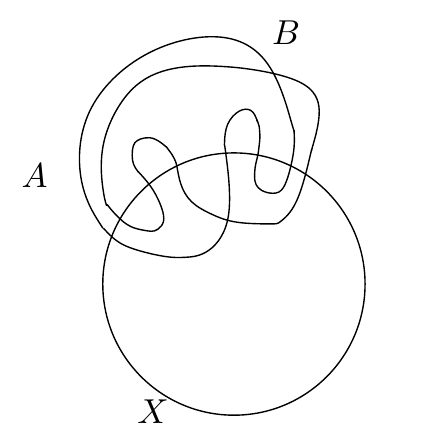}
  \captionof{figure}{Before $X$ pushes $A$ and $B$}
  \label{fig:threeobjects}
\end{minipage}%
\begin{minipage}{.5\textwidth}
  \centering
  \includegraphics[width=7cm]{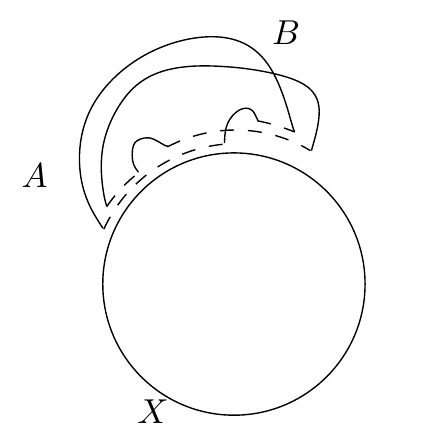}
  \captionof{figure}{After $X$ pushes $A$ and $B$}
  \label{fig:threeobjects2}
\end{minipage}
\end{figure}

\newpage
\section{Lower-bounds}
\label{appendix:lowerbounds}

In this section we give examples of regions of low (but superlinear) union complexity that do not admit a PTAS.

\begin{observation}
\label{observation:lowerbound}
The problem of approximating minimum-size set-cover is:
\begin{enumerate}
\item APX-hard for a set $\R$ of $n$ $4$-sided polygons in the plane
of union complexity $n \alpha (n)$.
\item inapproximable within $o(\log s)$ factor for a set $\R$ of $n$ $4s$-sided polygons in the plane
of union complexity $n 2^{\alpha(n)^{O(s)}}$, for any integer $s>3$.
\item inapproximable within $o(\log d)$ factor for a set $\R$ of $n$ halfspaces
in $\Re^d$, for any integer $d>3$.
\end{enumerate}
\end{observation}
\begin{proof}
\emph{1.} Chan-Grant~\cite{CG2014} showed that computing minimum size set-covers
for objects defined by shadows of line-segments in the plane is APX-hard. 
By Davenport-Schinzel sequences, the union complexity of $n$ line-segments in the plane 
is $O(n \alpha(n))$. These shadows can be `closed off' without any further
intersections to derive the $4$-sided polygons.

\emph{2.} Trevisan~\cite{Trevisan01} showed that computing minimum size set-covers
for general set systems $(V, \{S_1, \ldots, S_m\})$ where each $S_i$ has size $O(s)$
is inapproximable with factor $o(\log s)$ unless $P = NP$. 
These sets can be easily implemented using regions whose boundaries cross at most $O(s)$ times. To see this place a point corresponding to each vertex on the $x$-axis. Then for each set $S_i$ construct $x$-monotone curve $\gamma_i$ with $2s+1$ horizontal segments and $2s$ nearly vertical segments so that the points corresponding to the vertices in $S_i$ lie above $\gamma_i$ and all other points lies below $\gamma_i$. 
These curves can easily be drawn in such a way that any two of them intersect at most $O(s)$ times. Thus by bounds known on Davenport-Shinzel sequences, the lower envelope of the curves has complexity $O(n 2^{\alpha(n)^{O(s)}})$.
 Thus if we consider the regions $R_i$ defined by the set of points above $\gamma_i$ we get a set of regions with small union complexity. These regions can be
 made bounded without increasing the union complexity.

\emph{3.} There exist a set of points in $\Re^d$ (points on the moment curve; 
see Matousek~\cite{M02})
such that every $d/2$-sized subset can be obtained by intersection with a halfspace.
Thus a set-cover problem where every set has size  at most $d/2$ can be realized
with halfspaces in $\Re^d$, which together with the Trevisan bound~\cite{Trevisan01}
implies the lower-bound.
\end{proof}

\newpage

\section{Separators. Proof of Theorem~\ref{thm:weightedseparator}
and Theorem~\ref{thm:separator}}
\label{appendix:separators}

In this section we prove Theorem~\ref{thm:weightedseparator}
and Theorem~\ref{thm:separator}.
The proof of both these statements follow from a suitable
subdivision of the plane, and the application of a
variant of the planar graph separator theorem.
Our proof can also be seen as a generalization of the
separator theorem of Fox-Pach~\cite{FP11} where, given a set
of curves with $m$ intersections, they show the existence
of a separator that intersects $O(\sqrt{m})$ curves: this is obtained
by applying the planar graph separator theorem on the arrangement
induced by these curves (where each intersection is taken as a vertex).
We   also apply the planar graph separator theorem, but instead on
a coarser subdivision of the plane.
This subdivision is similar to a structure
for the case of lines  in the plane,
called \emph{cuttings}~\cite{M92}.


\begin{lemma}
\label{thm:mpartition}
Given a set $S$ of $n$ $x$-monotone curves in the plane  with $m$ intersections (and
where every pair of curves intersect $O(1)$ times),
and a parameter $r$, there exists a partition
of $\Re^2$ into $O(r + \frac{mr^2}{n^2})$ regions (each
of constant descriptive complexity, defined by a constant number
of curves of $\Gamma(S)$ together with constant number of vertical line-segments)
such that the interior of any region in this partition intersects
curves of total weight $O(n/r)$. 
\end{lemma}
{\bf Proof sketch.} We present the proof of the above lemma  in Appendix~\ref{appendix:mpartition}. 
Here we observe that a near-optimal (within log factors) result follows immediately from $\eps$-nets
(the proof in Appendix A gets rid of these log factors using standard techniques from the theory of $\eps$-nets). For the purpose of designing QPTAS, however, it is not necessary to get rid of the log factors and the near-optimal bounds suffice at the expense of a slightly higher, but still quasi-polynomial, running time.

Given $S$, consider the set-system $(S, \F)$ induced
by intersection with  segments in the plane, i.e.,
$ F \in \F  \text{ iff there exists
a line segment $l$ s.t. } F = \{ s \in S \ | \ s \cap l \neq \emptyset\}. $
This set system can be shown to have a finite VC dimension~\cite{AS40}.
Pick a random set $R$ by uniformly adding each curve of $S$
with probability $p = (C r \log r)/n$, where $C$
is a large constant. Then
 $R$ is a $(1/r)$-net for $(S, \F)$ with probability at least $9/10$~\cite{M02}.
The  expected size of $R$ is $np$,
and the expected number of intersections of curves  in $R$ is $mp^2$.
By Markov's inequality, with probability at least $9/10$,
the size of $R$ is at most $10np$,
and the number of intersections in $R$ is at most $10mp^2$.
Therefore with probability at least $8/10$, $R$
is a  $(1/r)$-net \emph{and} the size of the trapezoidal decomposition
of $R$ is $O(r\log r + (mr^2 \log^2 r)/n^2)$.
Note that any  open line-segment $l$ in this trapezoidal decomposition
must intersect at most $n/r$ segments of $S$, as otherwise
the set of curves intersecting $l$ would not be 
hit by a curve from $R$, contradicting the fact
that $R$ is a $(1/r)$-net.

\subsection*{Proof of Theorem~\ref{thm:weightedseparator}}

Let $\R = \{R_1, \ldots, R_n\}$ be the set of $n$ weighted
$\alpha$-simple regions, where weight of $R_i$ is $w_i$,
and $R_i$ can be decomposed into $v_i = v(R_i) \leq \alpha$ $x$-monotone curves.
In the standard way, by scaling, one can assume the weights are integral.
The total weight is $W = \sum_i w_i$. 
First decompose each region $R_i$ into
$v_i$ $x$-monotone curves of weight $w_i$ (and let $S$ be the set of all
such curves), and then each of these $v_i$ curves for region $R_i$
is replaced by $\alpha w_i/v_i$ copies of weight $1$
to get the set $S'$. Note that the total 
number of curves in $S'$ from region $R_i$ is $\alpha w_i$,
and $|S'| = \sum_i \alpha w_i = \alpha W$.
Apply Lemma~\ref{thm:mpartition} to $S'$ 
with the parameter $r$ to be set later; if the curves
in $S'$ are disjoint, the proof of
Lemma~\ref{thm:mpartition} shows the existence
of a subset $R'$ of $S'$ of total size $O(r)$  such that
the trapezoidal decomposition of $R'$ gives a partition $\T$
where the interior of any region in $\T$
intersects $O(\alpha W/r)$ curves of $S'$. Remove copies of the same curve
in $R'$ to get a subset $R$ of $S$ of total size $O(r)$. 
Now replace each   $s \in R$   with
a small expanded copy of $s$ which contains $s$ in its interior and modify $\T$ accordingly
to use edges of this new region instead of $s$. 
As each curve $s \in S$ of weight $\Omega(\alpha W/r)$ must be present
in $R$ (otherwise it would intersect the interior of some
region of $\T$ contradicting the partitioning property), this ensures that the
new triangulation $\T'$ has the properties
that $i)$ every $s \in S$ with weight $\Omega(\alpha W/r)$ lies in the interior
of a single face of $\T'$,
$ii)$ the number of vertices of $\T'$ is $O(r)$, and 
$iii$) each edge
in $\T'$  intersects curves in $S$ of total weight $O(\alpha W/r)$
(for small-enough replacing regions around each $s \in R$).

$\T'$ can be seen as an embedding of an underlying planar graph $G$.
Give weights to each face of $\T'$: 
if a curve $s \in S$ intersects $t$ faces of $\T'$,
add weight $w(s)/t$ to the weight of each of these $t$ faces.
A variant of the planar graph separator theorem~\cite{M86} now implies
the existence of a simple cycle $\C$ in $\T'$ of $O(\sqrt{r})$ vertices such that 
faces completely inside (and outside) have total weight at most $2 \alpha W/3$, and hence so do the curves of $S$ inside (and outside) $\C$.
Let $\R_{int} \subset \R$ be set of regions completely inside $\C$.
Then as $R_i$ produced curves of total weight $\alpha w_i$, we have
$\sum_{R_i \in \R_{int}} \alpha w_i \leq 2\alpha W/3$, implying
that stated bound on the total weight of regions in $\R$ completely inside $\C$, 
$\sum_{R_i \in \R_{int}} w_i \leq 2W/3$.
Finally, the weight of the regions of $\R$ intersected by $\C$ is at most 
$O(\sqrt{r}) \cdot O(\alpha W/r) = O(\alpha W/\sqrt{r})$. Setting $r = \alpha^2/\delta^2$
concludes the upper-bound.

Note that $\C$ consists of alternating pieces of curves of $R$
and vertical line-segments from the trapezoidal decomposition of $R$. 
Each vertical line-segment in this decomposition is between an
endpoint of a curve of $S$ and another curve  of $S$. Thus $\C$
can be specified by giving a sequence of $O(\alpha/\delta)$ curves
of $S$ together with specifying which endpoint is used
for every set of consecutive curves in the sequence.
These specifications require an additional $O(\alpha/\delta)$ bits.
The  optimality of this statement can be seen by the following
  construction where $\R$ consists of a set of disjoint line
  segments of weight $1$.
Take a regular polygon $\P$ with $c/\delta$ vertices, and place $\delta n/c$ copies of $\P$ concentrically, each shrunk slightly more than the previous one so that there are no intersections between any two copies. Note that one can choose 
the scaling factor small-enough such that
any closed curve separating two different copies of $\P$ 
must also have at least $c/\delta$ vertices.
Finally replace each polygon with $c/\delta$
line segments corresponding to its sides (slightly perturbed so that they are disjoint). 
Take any balanced closed curve $\C'$ in the plane. If it contains at least one
copy of $\P$ completely inside, and one copy completely outside, 
then by construction it has at least $c/\delta$ vertices. Otherwise,
say there is no copy of $\P$ completely inside $\C'$. As $\C'$ is balanced,
it contains at least $n/3$ curves inside or intersecting
its boundary; these curves
belong to at least $(n/3)/(c/\delta) = \delta n/(3c)$ different copies of $\P$,
and each of these copies must intersect $\C'$ in at least one curve.

\subsection*{Proof of Theorem~\ref{thm:separator}}

Given the set $\R$ of $n$ $\alpha$-simple regions with $m$ intersections,
construct from it the set $S$ of $O(\alpha n)$ $x$-monotone curves.
Apply Lemma~\ref{thm:mpartition} 
to $S$ to get
a partition $\T$ of $\Re^2$ into $O(r+mr^2/\alpha^2 n^2)$ regions.  
$\T$ can be seen as an embedding of an underlying planar graph $G$.
Give weights to each face of $\T$: 
if a region $s \in \R$ intersects $t$ faces of $\T$,
add weight $1/t$ to the weight of each of these $t$ faces.
Now from~\cite{M86} we get a simple cycle $\C$ in $\T$ of $O(\sqrt{r + mr^2/\alpha^2 n^2})$ vertices such that 
faces completely inside (and outside) have total weight at most $2n/3$, and hence so do
the regions of $\R$ inside (and outside) $\C$.
The weight of the regions of $\R$ intersected by $\C$ is at most 
$O(\sqrt{r + mr^2/\alpha^2 n^2}) \cdot O(\alpha n/r) = O(\sqrt{m+\alpha^2 n^2/r})$.

\newpage

\section{Proof of Lemma~\ref{thm:mpartition}}
\label{appendix:mpartition}

\begin{proof}
We briefly now review the basic 
partitioning method of using \emph{trapezoidal decompositions}.
Given a set $R \subseteq S$ of $x$-monotone curves, one can
 partition the space (say inside a large-enough rectangle
containing all the curves of $S$) as follows.
For each
endpoint of a curve in $R$ or an intersection-point between curves
in $R$, shoot a vertical ray upwards (and downwards) till it
hits another curve  (or the bounding rectangle). 
The union of all these vertical segments together with $R$
partitions the bounding rectangle into a set of regions.
A crucial fact is that each region $\Delta$ in this partition
 is determined by a constant
($2$, $3$ or $4$) number of curves in $R$.
Call such regions \emph{trapezoidal regions} (or trapezoids for brevity),
and the partition is called a \emph{trapezoidal decomposition}\footnote{We refer the reader to~\cite{M02} for a nice exposition
on trapezoidal decompositions.}.
Denote by $\Xi(R)$ this set of trapezoidal regions
in the trapezoidal decomposition of $R$.
The size, $|\Xi(R)|$, of the trapezoidal decomposition of $R$ is the
number of trapezoids  in  $\Xi(R)$; it is, within a constant-factor, 
equal to the total number of end- and intersection- points in $R$.
A trapezoid present in the trapezoidal decomposition
of any subset $R$ of $S$ is called a \emph{canonical trapezoid}.
For a canonical trapezoid $\Delta$, let $|\Delta|$ denote the set of curves
of $S$ intersected by $\Delta$.
 A trapezoid $\Delta$
is present in the trapezoidal decomposition of $R$
if and only if its determining curves are present in $R$,
and $R$ does not contain any of the curves of $S$ that intersect $\Delta$.
 For the rest of the proof, we only
work with canonical 
trapezoids determined by $4$ curves. The case for canonical
trapezoids determined by $2$ and $3$ curves is similar.

First note that a slightly weaker bound (within
logarithmic factors) follows immediately from $\eps$-nets.
Given $S$, consider the set-system $(S, \F)$ induced
by intersection with  segments in the plane, i.e.,
$$ F \in \F  \text{ iff there exists
a line segment $l$ s.t. } F = \{ s \in S \ | \ s \cap l \neq \emptyset\} $$
Pick a random set $R$ by uniformly adding each curve of $S$
with probability $p = (C r \log r)/n$, where $C$
is a large constant. Then
 $R$ is a $(1/r)$-net for $(S, \F)$ with probability at least $9/10$~\cite{M02}.
The  expected size of $R$ is $np$,
and the expected number of intersections of curves  in $R$ is $mp^2$.
By Markov's inequality, with probability at least $9/10$,
the size of $R$ is at most $10np$,
and the number of intersections in $R$ is at most $10mp^2$.
Therefore with probability at least $8/10$, $R$
is a  $(1/r)$-net \emph{and} the size of the trapezoidal decomposition
of $R$ is $O(r\log r + (mr^2 \log^2 r)/n^2)$.
Note that any  open line-segment $l$ in this trapezoidal decomposition
must intersect at most $n/r$ segments of $S$, as otherwise
the set of curves intersecting $l$ would not be 
hit by a curve from $R$, contradicting the fact
that $R$ is a $(1/r)$-net. 

Set $p = Cr/n$ (for a small-enough constant $C$ to be set later), and pick each  curve in $S$
with probability $p$ to get a random sample $R$.
Construct the trapezoidal decomposition $\Xi(R)$ of $R$.
If all trapezoids $\Delta \in \Xi(R)$ intersect at most $n/r$ curves
in $S$, we are done. Otherwise we will further partition
each violating $\Delta$, based on two ideas.
First, the expected number of trapezoids in $\Xi(R)$
intersecting more than $n/r$ curves are few.
In particular, we will show  (Lemma~\ref{alemma:expo}) that
the expected number of trapezoids intersecting at least $tn/r$ 
curves in $S$ is \emph{exponentially} decreasing as a function
of $t$. Second, consider a $\Delta$ intersecting a set, say $S_{\Delta}$, of 
$n_{\Delta} = tn/r$ curves of $S$.
Use the weaker bound on $S_{\Delta}$ with parameter $t$ to get a
partition inside $\Delta$ of $O(t\log t + (m_{\Delta} t^2 \log^2 t)/n_{\Delta}^2) 
= O(t^2 \log^2 t)$ trapezoids. By definition, each such
trapezoids intersects at most $n_{\Delta}/t$ = $n/r$ curves of $S_{\Delta}$ 
(and hence of $S$). Thus refining each $\Delta$ gives the required
partition on $S$ with parameter $r$. It remains to bound the overall 
expected size of this partition.

\begin{lemma} 
\label{alemma:cssegments}
Given a set $S$ of $n$ $x$-monotone curves in the plane with $m$ intersections,
the number of canonical trapezoids defined by $S$ that intersect
at most $k$ curves of $S$ is $O(nk^3 + mk^2)$.
\end{lemma}
\begin{proof}
Let $\Xi_{\leq k}$ be the set of canonical trapezoids defined by $S$ that intersect
at most $k$ curves of $S$.
The proof is standard via the Clarkson-Shor technique.
Construct a sample $T$ by adding each curve 
of $S$ with probability $p_0$; the expected total
number of picked curves is $np_0$ and the expected number
of intersections between the curves of $T$ is $mp_0^2$. 
The trick is to count the expected 
size of $\Xi(T)$ in two ways.
On one hand, it is at most $O(np_0 + mp_0^2)$ (i.e., the 
expected number of vertices present in $\Xi(T)$). 
On the other hand, as the probability of a canonical trapezoid
$\Delta$ being in $\Xi(T)$ is $p_0^4(1-p_0)^{|\Delta \cap S|}$, it is at least 
$$ \sum_{\Delta \in \Xi_{\leq k}} p_0^4 (1-p)^{|\Delta \cap S|} \geq
\sum_{\Delta \in \Xi_{\leq k}} p_0^4 (1-p_0)^{k}$$
where the sum is over all canonical trapezoids 
$\Delta$ which intersect at most $k$ curves of $S$.  Therefore, 
\begin{eqnarray*}
\sum_{\Delta} p_0^4 (1-p_0)^{k} = & |\Xi_{\leq k}| & \cdot \ p_0^4 (1-p_0)^{k} \leq E[|\Xi(T)|] =np_0 + mp_0^2 \\
& |\Xi_{\leq k}| & \leq \frac{np_0 + mp_0^2}{p_0^4(1-p_0)^k} = O(nk^3 + mk^2)
\end{eqnarray*}
for $p_0=1/2k$.
\end{proof}

\begin{lemma} 
\label{alemma:expo}
Expected number of trapezoids in $\Xi(R)$ intersecting
at least $tn/r$ curves of $S$ is  
$$O \left( (t^3 r + \frac{mr^2t^2}{n^2}) e^{-t} \right)$$ 
\end{lemma} 
\begin{proof}
By definition:
$$ E[ |\Delta \in \Xi(R) \ s.t. \ |\Delta \cap S| = tn/r| ]
= |\Delta \ s.t. \ \ |\Delta \cap S| = tn/r| \cdot p^4 (1-p)^{tn/r} $$

Using Lemma~\ref{alemma:cssegments},
\begin{eqnarray*}
E[ |\Delta \in \Xi(R) \ s.t. \ |\Delta \cap S| = tn/r| ] & \leq & O\left( n(tn/r)^3 + m(tn/r)^2 \right) p^4 (1-p)^{tn/r} \\
& = & O\left( (t^3 r + \frac{mr^2t^2}{n^2})e^{-t} \right)
\end{eqnarray*}
The bound follows by summing up over all trapezoids intersecting at least $tn/r$
curves in $S$.
\end{proof}

Now we can complete the proof of the theorem. Let $n_{\Delta} = t_{\Delta} n/r$
be the number of   curves in $S$ intersected by
each trapezoid $\Delta \in \Xi(R)$ (and $m_{\Delta}$ the number of their intersections).
Using the weaker bound, refine
trapezoid $\Delta$ by adding a  $(1/t_{\Delta})$-net $R_{\Delta}$
for all the $t_{\Delta} n/r$ curves of $S$ intersected by $\Delta$.
The resulting expected total size of the trapezoidal partition is:
\begin{eqnarray*}
&=&  |R| + \sum_{\Delta} Pr[\Delta \in \Xi(R)] \cdot \text{ Size of trapezoidal decomposition of   $(1/t_{\Delta})$-net within $\Delta$} \\
&=& |R| + \sum_{\Delta} Pr[\Delta \in \Xi(R)] \cdot O\left( t_{\Delta} \log t_{\Delta} + \frac{m_{\Delta}t_{\Delta}^2 \log^2 t_{\Delta}}{n_{\Delta}^2} \right) \text{ (using the weaker bound)} \\
&\leq&  |R| + \sum_{\Delta} Pr[\Delta \in \Xi(R)] \cdot  O\left( t_{\Delta}^2 \log^2 t_{\Delta} \right) \text{ (as $m_{\Delta} \leq n_{\Delta}^2$)}\\
&=& |R| + \sum_j \sum_{\substack{\Delta \ s.t. \ \\ 2^j \leq t_{\Delta} \leq 2^{j+1}}} Pr[\Delta \in \Xi(R)] \cdot  O\left( t_{\Delta}^2 \log^2 t_{\Delta} \right) \\
&\leq& |R| + \sum_j E[ \text{ Number of trapezoids $\Delta$ in $\Xi(R)$ with $2^j \leq t_{\Delta}$ } ] \cdot O\left( 2^{2(j+1)} \log^2 2^{j+1} \right) \\
&\leq& |R| + \sum_j O\left( (2^{3j} r + \frac{mr^22^{2j}}{n^2}) e^{-2^j} \right) \cdot O\left( 2^{2(j+1)} \log^2 2^{j+1} \right) \text{ (Lemma~\ref{alemma:expo})}\\
&=& |R| + r \sum_j O\left( 2^{3j}  e^{-2^j} \right) \cdot O\left( 2^{2(j+1)} \log^2 2^{j+1} \right) +  \frac{mr^2}{n^2} \sum_j O\left( 2^{2j} e^{-2^j} \right) \cdot O\left( 2^{2(j+1)} \log^2 2^{j+1} \right)\\
&=& np + mp^2 + O(r) + O(\frac{mr^2}{n^2})  = O( r+ \frac{mr^2}{n^2}) \text{ (the summands form a geometric series)}
\end{eqnarray*}
as required. This finishes the proof of Lemma~\ref{thm:mpartition}.
\end{proof}

\newpage

\end{document}